\documentclass{article}
\usepackage[utf8]{inputenx}
\usepackage{graphicx} % Required for inserting images
\usepackage{indentfirst}
\usepackage{microtype}
\usepackage{yfonts,color}
\usepackage{minted}

\usepackage{authblk}
\usepackage{comment}
\usepackage{subfig}
\usepackage{amssymb}
\usepackage{algorithm}
\usepackage{algpseudocode}
\usepackage{amsfonts}
\usepackage{amsmath,bm}

\usepackage{enumitem}
\numberwithin{equation}{section}
\usepackage{lettrine}
\usepackage{amsthm}
\usepackage{multirow}
\usepackage{hyperref}

\newcommand{\bolda}{\tilde{\bm{a}}}
\newcommand{\minCZ}{\underline{C}_{\, \ell;Z}}
\newcommand{\maxCZ}{\overline{C}_{\ell;Z}}
\newcommand{\phiproxi}{\bm{\phi}^*}
\newcommand{\qproxi}{q^*}
\newcommand{\LN}{L}
\newcommand{\FN}{\mathcal{F}_{\ell}}
\newcommand{\lambdaN}{\lambda_{\ell}}
\newcommand{\qN}{p}

\newcommand{\phihat}{\bm{\widehat{\phi}}}

\newcommand{\brutti}{C_L}
\newcommand{\Gammastar}{\Gamma^*_{\ell}}

\newtheorem{theorem}{Theorem}[section]

\newtheorem{condition}[theorem]{Condition}

\newtheorem{definition}[theorem]{Definition}

\newtheorem{lemma}[theorem]{Lemma}

\newtheorem{proposition}[theorem]{Proposition}
\newtheorem{remark}[theorem]{Remark}

\counterwithin{algorithm}{section}

\usepackage{chngcntr}
\usepackage{cite}
\usepackage{balance}
\usepackage[nottoc, notlof, notlot]{tocbibind}
\usepackage{hyperref}
\usepackage{mathrsfs}
\usepackage{listings}
\usepackage{xcolor}
\usepackage{pythonhighlight}
\definecolor{codegreen}{rgb}{0,0.6,0}
\definecolor{codegray}{rgb}{0.5,0.5,0.5}
\definecolor{codepurple}{rgb}{0.58,0,0.82}
\definecolor{backcolour}{rgb}{0.95,0.95,0.92}

\title{Change Point Detection for Functional Autoregressive Processes on the Sphere}
\author{Federica Spoto\thanks{Department of Biostatistics, Harvard T.H. Chan School of Public Health, Boston, USA; e-mail: \href{mailto:fspoto@hsph.harvard.edu}{fspoto@hsph.harvard.edu}}, Alessia Caponera\thanks{Department of AI, Data and Decision Sciences, LUISS Guido Carli, Rome, Italy; e-mail: \href{mailto:acaponera@luiss.it}{acaponera@luiss.it}}, Pierpaolo Brutti\thanks{Department of Statistical Sciences, Sapienza Universit\`a di Roma, Rome, Italy; e-mail: \href{mailto:pierpaolo.brutti@uniroma1.it}{pierpaolo.brutti@uniroma1.it}}}

\date{}

\begin{document}

\maketitle

\section{Introduction}

Understanding and modeling changes or transitions in the behavior of natural and social phenomena is a fundamental aspect across scientific disciplines, encompassing domains such as climate science, economics, and medicine (see \cite{Chen2012}). These changes may be abrupt or gradual, arising from a multitude of mechanisms, and their timely detection is vital for uncovering the driving forces underlying dynamic processes. For instance, a shift in climatic conditions can have significant repercussions on ecosystems, agriculture, and human societies, while sudden transitions in investor sentiment can profoundly influence financial markets. Detecting and analyzing such changes affords researchers a deeper comprehension of the temporal and possibly spatial evolution inherent to complex systems.

The statistical framework of change point analysis identifies and characterizes abrupt changes in the generative mechanisms of observed data, to address applications ranging from signal processing to finance and environmental monitoring \cite{Burg2020AnEO, Truong2020, fearnhead2022detecting}. The core objective is to detect points in time where the underlying probabilistic structure of a process changes, often manifesting as shifts in distribution, mean, variance, or model parameters \cite{Kim2009, Tartakovsky2008, Reeves2007, Talih, basseville1993detection, adams2007bayesian}. Foundational work includes classical methods like the Likelihood Ratio (LR) test \cite{Hinkley, Sen1975, Horvath1993, Chen1997} and the Cumulative Sum (CUSUM) test \cite{Page1955, caponera2024multi}, which have provided a basis for later developments. For data exhibiting multiple change points, recursive partitioning algorithms such as Binary Segmentation \cite{Scott1974, vostrikova1981detecting}, Wild Binary Segmentation \cite{Fryzlewicz2014}, and dynamic programming techniques including PELT \cite{Killick} are widely used, offering a balance of computational efficiency and statistical rigor.

For time series characterized by temporal dependence, particularly autoregressive (AR) and ARMA models, the presence of parameter changes violates assumptions of stationarity central to classical analysis \cite{brockwell1991time, Robbins2016, fearnhead2022detecting}. This has led to specialized methodologies for detecting changes in autoregressive coefficients or innovation variances \cite{Inclan1994, Trapani2022}, as well as for more complex structures such as vector autoregressive (VAR) processes (\cite{Sims80, Var}). Penalized estimation methods, including fused lasso and total variation regularization (\cite{Safikhani, Chan2014}) and penalized dynamic programming \cite{Localizing_ar}, have been successfully employed to simultaneously identify change points and estimate segment-specific model parameters, even in high-dimensional contexts \cite{ChoFryzlewicz, Cho2016, WangSamworth2018, Safikhani}. Bayesian techniques \cite{Chernoff1964, Carlin1992, ALTIERI2015197} and universal inference frameworks \cite{larry_universal, spoto2022universal} further enhance robustness, accommodating model misspecification and finite-sample uncertainty.

Despite the maturity of change point analysis for classical Euclidean time series, analogous methods for non-Euclidean data -- in particular, for spherical or manifold-valued autoregressive processes -- remain largely undeveloped. This gap is especially notable given the growing importance of random fields on the sphere in areas ranging from cosmology and astrophysics to geophysics and climate research \cite{baldi2009, Berg2016}. Spherical random fields frequently exhibit both spatial and temporal dependencies, with models for their analysis ranging from purely spatial constructions to spatio-temporal and covariance-based approaches. Of particular relevance is the class of spherical functional autoregressive models (SPHAR($p$)), as introduced in \cite{caponera2019asymptotics}, which provides a flexible framework for capturing the dual spatio-temporal dependence structure inherent in spherical-valued data. Notably, however, these models assume stationarity -- an assumption that may be critically violated in applications where long-term regime changes or abrupt shifts are expected.

To the best of our knowledge, no prior work has developed a general change point methodology for functional autoregressive models on the sphere. A related contribution is \cite{caponera2024multi}, which proposes a CUSUM test statistic to assess non-stationarity in the mean structure across a wide range of alternatives, including abrupt changes.  
As a result, in many domains of application -- such as the detection of changes in global temperature patterns -- the existing literature either employs purely temporal analyses or ignores the spatial component altogether, potentially missing important structural breaks and their spatial manifestations (see, for example, \cite{Reeves2007} and references therein).

The primary objective of this paper is to address this methodological gap by generalizing the SPHAR($p$) model to accommodate non-stationarity via the integration of the change point framework. We moreover develop a penalized estimation procedure and provide theoretical guarantees for consistent detection and accurate localization of change points in this setting. 
Our approach generalizes classical penalized change point techniques to the setting of piecewise stationary, spherical functional autoregressive processes -- a fundamentally infinite-dimensional and non-Euclidean context. This allows joint detection and localization of regime shifts in the spatio-temporal structure of spherical random fields. This extension, by inherently synthesizing spatial and temporal information, is anticipated to offer novel insights into both global and localized changes in spherical random fields, thereby broadening the scope and applicability of change point methodologies in the context of manifold-valued time series.

The structure of the paper is as follows. Section \ref{sec:spharcp_gen} presents the spherical change-point autoregressive model, discusses its spectral characteristics, and outlines the main assumptions required for the subsequent theoretical developments.
In Section \ref{sec:est_main} we then present the methodology used to estimate the number of change points and their locations, and we state the main result on the consistency of the corresponding estimators. Section \ref{sec:theor_results} contains some additional theoretical results on the prediction error and the optimal partition. In Section \ref{sec:num_res} we show the predictive performance of the proposed procedure through simulations.
Finally in Appendix \ref{sec:est_singleinterval} and Appendix \ref{sec:proofs} we collect some auxiliary results and all the proofs.

\section{Model Specification and Assumptions}\label{sec:spharcp_gen}

\subsection{Model Components} \label{sec:def_model}

In this section we refer to several definitions introduced in \cite{caponera2019asymptotics} in order to define the $K+1$ processes building the piecewise-stationary functional autoregressive model with function domain the unit sphere $\mathbb{S}^2$. In particular, we deal with 2-weakly \emph{isotropic-stationary processes} $$\{T_{k}(x, t), (x,t) \in \mathbb{S}^2\times \mathbb{Z}\}, \qquad k = 0, \dots, K,$$ as defined in Section 3 in \cite{caponera2019asymptotics}, independent of each other, and \emph{spherical white noise processes} (Definition 3 in \cite{caponera2019asymptotics}) $$\{Z_{k}(x, t), (x,t) \in \mathbb{S}^2\times \mathbb{Z}\}, \qquad k = 0, \dots, K.$$
All the processes are defined on a common probability space $(\Omega, \mathscr{F}, P)$.

Let $L^2(\mathbb{S}^2) = L^2(\mathbb{S}^2, dx)$ be the space of real-valued square-integrable functions on the sphere, $dx$ being the spherical Lebesgue measure, and let $\{Y_{\ell,m}: \ell \ge 0, m = -\ell, \dots, \ell\}$ be the standard orthonormal basis of real spherical harmonics for $L^2(\mathbb{S}^2)$.

We model the $T_k$'s as spherical functional autoregressive processes of order $p$ (SPHAR($p$)) with innovation processes the $Z_k$'s (see also Chapters 3 and 5 in \cite{bosq}). This means that for the generic $k-$th process there exist $p$ linear and bounded operators $\Phi_{k,j}: L^2(\mathbb{S}^2) \to L^2(\mathbb{S}^2)$, $j=1,\dots,p$, such that 
\begin{equation}
    T_{k}(x,t) = \sum_{j = 1}^p(\Phi_{k,j} T_{k, t-j})(x) + Z_{k}(x,t) ,\qquad (x,t) \in \mathbb{S}^2\times \mathbb{Z},
    \label{eq:model}
\end{equation}
the equality holding in $L^2(\Omega)$ and $L^2(\Omega \times \mathbb{S}^2)$. In particular, the autoregressive kernel operator $\Phi_{k,j}$ is defined as
$$
(\Phi_{k,j} f)(x) = \int_{\mathbb{S}^2} h_{k,j}(\langle x,y \rangle)f(y)dy, \qquad x \in \mathbb{S}^2,
$$
for some continuous and isotropic autoregressive kernel $h_{k,j}:[-1,1] \rightarrow \mathbb{R}$, $j=1,\dots,p$ (see also Definition 5 in \cite{caponera2019asymptotics}). 
Under isotropy, the random fields $T_k$ and $Z_k$, $k = 0, \dots, K$, admit the expansions
$$T_{k}(x,t) =\sum_{\ell=0}^{\infty} \sum_{m=-\ell}^{\ell} a^{(k)}_{\ell,m}(t) \cdot Y_{\ell,m}(x), \qquad (x,t) \in \mathbb{S}^2\times \mathbb{Z},$$
$$Z_{k}(x,t) =\sum_{\ell=0}^{\infty} \sum_{m=-\ell}^{\ell} a^{(k)}_{\ell,m;Z}(t) \cdot Y_{\ell,m}(x), \qquad (x,t) \in \mathbb{S}^2\times \mathbb{Z},$$
in $L^2(\Omega)$ and $L^2(\Omega \times \mathbb{S}^2)$.
As a consequence, for $\ell \ge 0,\  m = -\ell, \dots, \ell, \ t \in \mathbb{Z},$ the random harmonic coefficients, defined as
\begin{equation}\label{eq:alm_def}  
    a^{(k)}_{\ell,m}(t) = \int_{\mathbb{S}^2} T_k(x,t) Y_{\ell,m}(x)dx,\qquad 
    a^{(k)}_{\ell,m;Z}(t) = \int_{\mathbb{S}^2} Z_k(x,t) Y_{\ell,m}(x)dx,
\end{equation}
satisfy the $p-$th order autoregressive equation 
\begin{equation} \label{eq::ar(p)}
    a_{\ell, m }^{(k)}(t) = \sum_{j = 1}^p\phi^{(k)}_{\ell;j} a_{\ell, m }^{(k)}(t-j) + a_{\ell, m; Z}^{(k)}(t)
\end{equation}
with $\{\phi^{(k)}_{\ell;j},\ \ell \ge 0\}$ being the eigenvalues of the operator $\Phi_{k,j}$. Indeed, it holds that
$$
\Phi_{k,j} Y_{\ell,m}= \phi^{(k)}_{\ell;j}Y_{\ell,m}.
$$
Moreover, for any $\ell, \ell' \in  \mathbb{N}, \ m = -\ell, \dots, \ell, \ m' = - \ell', \dots, \ell'$, and $k, k' = 0, \dots, K,$ we have that 
$$
\mathbb{E}[a^{(k)}_{\ell,m}(t)a^{(k')}_{\ell',m'}(t+\tau)] = \begin{cases}
    C_{\ell}^{(k)}(\tau) &\text{if } \ell = \ell', m=m', k=k',\\
    0 & \text{otherwise}.
\end{cases} \qquad t, \tau \in \mathbb{Z}.
$$
Observe that, for $\tau = 0$, $\{C_{\ell}^{(k)}(0), \ \ell \ge 0\}$ corresponds to the so-called angular power spectrum, the spectral decomposition of the covariance function of a purely spatial spherical random field. From now on, we will use the notation $C_{\ell}^{(k)}(0) = C_{\ell}^{(k)}$.

In addition, we define $\{C_{\ell;Z}^{(k)}, \ell \ge 0\}$, which is the power spectrum of the spherical white noise $\{Z_k(x,t), \ (x,t) \in \mathbb{S}^2\times \mathbb{Z}\}$, and its maximum and minimum over the processes $$
\maxCZ = \underset{k= 0, \dots,  K}{\max} C_{\ell;Z}^{(k)}, \qquad \minCZ = \underset{k= 0, \dots,  K}{\min} C_{\ell;Z}^{(k)}.$$

\begin{condition}\label{cond:ratio_cl} The ratio between the maximum $\maxCZ$ and minimum $\minCZ$ is assumed to be square-summable in $\ell$, i.e.,
        $$\sum_{\ell= 0}^\infty \left|\frac{\maxCZ }{\minCZ }\right|^2< \infty.$$
\end{condition}
Note that, in this context, the latter condition highlights one of the differences with the finite dimensional setting. Indeed, while in the finite dimensional case this condition is automatically satisfied, in the infinite dimensional set up it is required, since $\minCZ \to 0$ as $\ell \to \infty$. This suggests that the noise variance at higher multipoles does not differ too much between the $K+1$ segments.
Moreover, this implies that 
$$
\sum_{\ell=0}^\infty \maxCZ (2\ell+1) \le \sup_{\ell\ge 0} \frac{\maxCZ }{\minCZ } \sum_{\ell=0}^\infty \minCZ (2\ell+1) \le \sup_{\ell\ge 0} \frac{\maxCZ }{\minCZ } \sum_{\ell=0}^\infty C_{\ell;Z}^{(k)} (2\ell+1) < \infty,
$$
for any $k=0,\dots,K$.
\\

For the sections to follow, it will be useful to introduce the $p-$dimensional vectors $\bm{\phi}_{\ell}^{(k)} = (\phi^{(k)}_{\ell;1}, \dots, \phi^{(k)}_{\ell;p})'$ and $\bolda_{\ell,m}^{(k)}(t) = (a_{\ell,m}^{(k)}(t), \dots, a_{\ell,m}^{(k)}(t-p+1))'$. The model for the random harmonic coefficients at given $(\ell,m)$ can be written as
\begin{equation*}\label{eq::armodelk}
    a^{(k)}_{\ell,m}(t) = \bm{\phi}_{\ell}^{(k)}{}' \bolda_{\ell,m}^{(k)}(t-1) + a^{(k)}_{\ell,m;Z}(t), \qquad t \in \mathbb{Z}.
\end{equation*}

\subsection{Spherical Change Point Autoregressive Model}\label{sec:obs_model}
Let $1=\eta_0 <\eta_1 <\dots<\eta_K <n<\eta_{K+1}=n+1$ be a sequence of change points. We denote the minimal spacing between consecutive change points as
\begin{equation}\label{eq:delta}
    \Delta = \underset{k= 0, \dots,  K}{\min}(\eta_{k+1} - \eta_{k}).
\end{equation}
We assume that $\Delta > p \ge 1$.

Consider a time-varying spherical random field $\{T(x,t),\ x \in \mathbb{S}^2, \ t \in \mathbb{Z}\}$. For $t \in \{1, \dots, n\}$, we assume to be able to observe a time-stretch of the spherical random field such that 
\begin{equation}\label{eq:obs_field}
   T(x,t) = T_k(x,t) \qquad \text{for } t \in [\eta_k, \eta_{k+1}) \cap \mathbb{Z},\ x \in \mathbb{S}^2. 
\end{equation}
Consequently, we assume to be able to observe for every timestamp $t \in \{1, \dots, n\}$ the corresponding harmonic coefficients $\{a_{\ell,m}(t), \ \ell \ge 0, \ m=-\ell,\dots,\ell\}$ (possibly up to a finite multipole $L>0$). The given coefficients are such that
$$
a_{\ell, m}(t) = a_{\ell, m}^{(k)}(t), \qquad \text{for } t \in [\eta_k, \eta_{k+1})\cap \mathbb{Z}.
$$
It will be useful later to consider the $p-$dimensional vectors 
$$
\bolda_{\ell,m}(t) = (a_{\ell,m}(t), \dots, a_{\ell,m}(t-p+1))', \qquad t \in \{p, \dots, n\}.
$$
Moreover, we can define the $p-$dimensional vectors $\bm{\phi}_\ell(t) = (\phi_{\ell;1}(t),\dots,\phi_{\ell,p}(t))'$ satisfying
$$
\bm{\phi}_\ell(t)  = \bm{\phi}^{(k)}_\ell, \qquad \text{for } t \in [\eta_k, \eta_{k+1}).
$$

The spherical change point process as described in \eqref{eq:obs_field} allows us to handle the non-stationarity of the spherical process $\{T(x,t),\ x \in \mathbb{S}^2, \ t \in \mathbb{Z}\}$ by modeling the process as the concatenation of $K+1$ different stationary spherical processes, where the transition happens at specific timestamps $1<\eta_1 <\dots<\eta_K <n$. 

\subsection{Assumptions}
We introduce some definitions and conditions necessary in the upcoming 
%to the model to ensure the properties and to achieve our 
theoretical results. The following definitions and assumptions refer to the generic $k-$th process but they are assumed to hold for all $ k \in \{0, \ \dots, K\}$. For additional details on the assumptions, see \cite{CAPONERA2021167}.

For a vector $v= (v_1,\dots,v_p)' \in \mathbb{R}^p$, $\|v\|_{2}$ indicates the Euclidean norm (or $2-$norm), while  $$\|v\|_{0}=\sum_{i=1}^p \mathbf{1}\{v_i\neq 0\}, \quad\|v\|_{1}=\sum_{i=1}^p|v_i|, \quad  \|v\|_{\infty} = \max_{i=1,\ldots,p} |v_i|,$$. We say that $v$ is a $r-$sparse vector, $1\le r \le p$, if $\|v\|_{0}=r$.

\begin{definition}[Sparsity set]\label{def:sparsity_set} For any $\ell \ge 0$ and $\bm{\phi}_{\ell}^{(k)}$, we define $q_{\ell}^{(k)} = \|\bm{\phi}_{\ell}^{(k)}\|_0$ the $\ell-$th sparsity index of segment $k$, which satisfies $0 \le q_{\ell}^{(k)} \le p$. We call $\{q_{\ell}^{(k)}:  \ell \ge 0, \ k = 0, \dots, K\}$ the sparsity set. 
\end{definition}

\begin{remark}
In order to ensure identifiability, we assume that there exists at least one $\ell \ge 0$ such that $\phi_{\ell;p}^{(k)}\neq 0$, so that $P(\Phi_{k,p}T_k(\cdot,t)\neq 0) >0$, for all $t \in \mathbb{Z}$, see \cite{bosq}. As a consequence, for
some $\ell \ge 0$, we can have $\bm{\phi}_{\ell}^{(k)}=0$ and hence $\|\bm{\phi}_{\ell}^{(k)}\|_0 = q_{\ell}^{(k)} = 0$; however, $\underset{\ell \ge 0 }{\max} \ q_{\ell}^{(k)} \ge 1.$     
\end{remark}

\begin{condition}[Gaussianity and identifiability]\label{cond:identifiability} Let $\{Z_k(x,t), \ (x,t) \in \mathbb{S}^2\times \mathbb{Z}\}$, $k=0,\dots, K$, be the spherical white noise processes in \eqref{eq:model}. Assume that they are Gaussian and $C^{(k)}_{\ell;Z} >0,$ $\forall \ell \ge 0$.
%$$
%\int_{\mathbb{S}^2 \times \mathbb{S}^2} \operatorname{Cov} \left(Z_k(x,t), Z_k(y,t)\right)\  f(x) \ f(y) \ dx \ dy >0,
%$$
%for any $f \in L^2(\mathbb{S}^2)$ such that $f(\cdot) \neq 0$. 
%Note that this is equivalent to ask that the covariance functions are strictly positive definite. 
\end{condition}
Now, for the generic $k-$th spherical process, we consider the equations associated with the subprocesses in \ref{eq::ar(p)}
\begin{equation}\label{eq::poly}
1-\phi^{(k)}_{\ell;1}z-\cdots-\phi^{(k)}_{\ell;p}z^p = 0, \qquad \ell \ge 0.
\end{equation}

\begin{condition}[Causality/Stationarity]\label{cond:causality}
For all $k \in \{0,\dots,K\}$, the roots of \eqref{eq::poly} lie outside the unit circle, i.e.,
$$
|z| \le 1\; \Rightarrow \; 1-\phi^{(k)}_{\ell;1}z-\cdots-\phi^{(k)}_{\ell;p}z^p \neq 0, \qquad \forall \ell \ge 0.
$$
\end{condition}
This condition ensures the existence of a unique isotropic and stationary solution for every process defined in \eqref{eq:model}; for rigorous proofs, see \cite{bosq,spharma20,tesi_ale}.
\begin{condition}[Smoothness]\label{cond:smoothness} The operators $\Phi_{k,j}$, $j=1, \ldots,p$, are nuclear, that is,
$$
\sum_{\ell \ge 0}(2\ell +1)|\phi^{(k)}_{\ell;j}| < \infty.
$$
\end{condition}
For a general discussion on operators on Hilbert spaces, see \cite{hsing}.
\begin{condition}[Boundedness]\label{cond:boundedness} For some absolute constant $C_{\Phi}>0$ it holds $$\max_{k=0,\dots,K}\sup_{\ell \ge 0}\|\bm{\phi}_\ell^{(k)}\|^2_2 \le C_{\Phi}.$$
\end{condition}

\subsection{Spectral Densities and Stability Measures}\label{sec:stability_meas}
In the following, we discuss local (to each process) and global stability measures for SPHAR($p$) random fields.
%from both considering each process separately and defining global measures over the set of processes. 
Stability measures analyze how stable the autocovariance matrix is, quantifying the dependency level between the variables of the process. \\

Let $k \in \{0, \dots, K\}$ and fix $(\ell,m)$.
Recall that $\{a^{(k)}_{\ell,m}(t),\ t \in \mathbb{Z}\}$ can be read as a real-valued autoregressive process of order $p$, see \eqref{eq::ar(p)}. Under standard stationarity assumptions (see \cite{brockwell1991time}), we can define its spectral density as
$$
f_{\ell}^{(k)}(\nu)= \frac{1}{2\pi}\sum_{\tau = -\infty}^{\infty} C_{\ell}^{(k)}(\tau)e^{-i\nu \tau} = \frac{1}{2\pi}\frac{C_{\ell;Z}^{(k)}}{|\bm{\phi}_{\ell}(e^{-i\nu})|^2}, \qquad\nu \in [-\pi, \pi],
$$
which is bounded and continuous. Upper and lower extrema of the spectral density over the unit circle are defined as
\begin{align*}
    & \mathcal{M}\left(f_{\ell}^{(k)}\right) = \underset{\nu \in [-\pi,\pi]}{\max}f_{\ell}^{(k)}(\nu), \qquad \mathfrak{m}\left(f_{\ell}^{(k)}\right) = \underset{\nu \in [-\pi,\pi]}{\min}f_{\ell}^{(k)}(\nu).
\end{align*}
Only referring to multipole $\ell$, $\mathcal{M}\left(f_{\ell}^{(k)}\right)$ can be considered as a stability measure of the process $\{a^{(k)}_{\ell,m}(t), t\in \mathbb{Z}\}$. Following \cite{CAPONERA2021167}, a global stability measure can be obtained by considering jointly all the segments via the following definition
\begin{align*}
    & \mathcal{M}_{\ell} = \underset{k= 0, \dots,  K}{\max}\,\mathcal{M}\left(f_{\ell}^{(k)}\right), \qquad
    %& \mathcal{M}_T = \underset{\ell \ge 0}{\max}\,\mathcal{M}_\ell. \\
    \mathfrak{m}_\ell = \underset{k= 0, \dots,  K}{\min}\,\mathfrak{m}(f_{\ell}^{(k)}).
\end{align*}

The following quantities are also well-defined both at process and global level
$$
\mu_{\min; \ell}(\bm{\phi}_\ell^{(k)})= \underset{z \in \mathbb{C}: |z| =1}{\min}|1-\phi^{(k)}_{\ell;1}z-\cdots-\phi^{(k)}_{\ell;p}z^p |^2,
$$
$$
\mu_{\max; \ell}(\bm{\phi}_\ell^{(k)}) = \underset{z \in \mathbb{C}: |z| =1}{\max}|1-\phi^{(k)}_{\ell;1}z-\cdots-\phi^{(k)}_{\ell;p}z^p |^2,
$$
$$
\mu_{\min; \ell}= \underset{k=0,\dots,K}{\min}\ \mu_{\min; \ell}(\bm{\phi}_\ell^{(k)}), \qquad \mu_{\max; \ell} = \underset{k=0,\dots,K}{\max} \ \mu_{\max; \ell}(\bm{\phi}_\ell^{(k)}).
$$

For every $k-$th spherical process, we define $\Gamma^{(k)}_{\ell}$ as the $p \times p$ matrix 
with generic $ij-$th element given by $C^{(k)}_\ell(i-j)$. Then, it holds
$$
2\pi \mathfrak{m}_\ell \le \min_{k = 0,\dots,K}\lambda_{\min} (\Gamma^{(k)}_\ell) \le \max_{k = 0,\dots,K} \lambda_{\max} (\Gamma^{(k)}_\ell) \le 2 \pi \mathcal{M}_\ell. 
$$
This is a general result that holds true for covariance matrices of any given dimension of the described type (see \cite{basu}). 

\section{Estimation and Main Result}\label{sec:est_main}

In this section, we introduce the proposed method to identify the number of change points and their locations, find the optimal data partition, and estimate the model parameters for each estimated segment. Specifically, following \cite{CAPONERA2021167, Rinaldo2020LocalizingCI}, we focus on a nested LASSO-type minimization problem which employs the so-called \emph{minimal partitioning problem}. We then state our main result which consists of a consistency result for the number of change points and their locations.

Let $\mathcal{P}$ be an interval partition of $\{1, \dots, n\}$ into $K_{\mathcal{P}}+1$ intervals, i.e.,
$$
\mathcal{P} = \{\{1, \dots, i_1-1\},\{i_1, \dots, i_2-1\}, \dots, \{i_{K_{\mathcal{P}}}, \dots, n\}\},
$$
for some integers $1= i_0 < i_1 < \dots < i_{K_{\mathcal{P}}} < n$, where $K_{\mathcal{P}}\ge 0$. Moreover, consider the generic interval $I = [s,e]$, where $[s,e]$ is a simplified notation for the set of consecutive integers $\{s, \dots, e \}$, with $s,e \in \{1,\dots,n\}$. 

\begin{definition}[Penalized estimator]\label{def:pen_estim}
    For a non-negative tuning parameter $\gamma\ge 0$, let
\begin{align}
    \widehat{\mathcal{P}} \in \underset{\mathcal{P}}{\mathrm{argmin} }\left \{\sum_{I\in \mathcal{P}}\mathcal{L}(I) + \gamma |\mathcal{P}|\right\}, \label{eq:min_problem}
\end{align}
where $\mathcal{L}(\cdot)$ is a loss function to be specified below, $|\mathcal P|$ is the cardinality of $\mathcal P$ and the minimization is taken over all possible interval partitions of  $\{1, \dots, n\}$.
Let $I = [s,e]$, $|I| = N_I = e - s - p + 1$, and $L>0$. Then, the dynamic programming algorithm can be deployed by setting the objective function to be \eqref{eq:min_problem} with
\begin{align}
\mathcal{L}(I) =
    \sum\limits_{t = s+p}^e  \sum\limits_{\ell = 0}^{\LN-1} \sum\limits_{m=-\ell}^{\ell}\left| a_{\ell, m}(t) - \bolda'_{\ell,m}(t-1) \bm\phihat_{\ell,I}  \right|^2,\label{eq:loss}
    \end{align}
%where
%$$
%(\widehat{\Phi}_{1}^{I},\dots, \widehat{\Phi}_{p}^{I})' = \sum_{\ell =0 }^{\LN-1} \bm{\widehat{\phi}}_{\ell, I}^{I} Y_{\ell,m} \otimes Y_{\ell,m},
%$$
%and 
where
\begin{align}
     \phihat_{\ell,I}&= \underset{\bm{\phi}_{\ell}
     \in \mathbb{R}^p}{\mathrm{argmin}} \ \sum_{t = s+p}^e  \sum_{m=-\ell}^{\ell}\left| a_{\ell, m}(t) - \bolda'_{\ell,m}(t-1) \bm{\phi}_{\ell}  \right|^2 + \lambda_\ell \sqrt{N_I(2\ell+1)}\|\bm{\phi}_{\ell}\|_1, \label{eq:phi_est}
\end{align}
with penalty parameter $\lambda_\ell \ge 0.$
\end{definition} 

The optimization problem in \eqref{eq:min_problem}, known as the minimal partitioning problem, can be solved using dynamic programming. The parameter $\gamma$ penalizes over-partitioning, while the scalar $\lambda$ is a tuning parameter controlling the norm of the spectral parameters and may depend on the multipole $\ell$. In \eqref{eq:loss} and \eqref{eq:phi_est}, the summation over timestamps is restricted to $t = s + p, \dots, e$, so that the loss function $\mathcal{L}(\cdot)$ for each interval $I$ depends exclusively on the data within that interval. For the estimation of the autoregressive spectral parameters when no change point occurs, see \cite{CAPONERA2021167}.

\begin{remark}
The estimators of the spectral parameters in \eqref{eq:phi_est} differ from those presented in \cite{CAPONERA2021167}. In particular, in \cite[Equation (3.1)]{CAPONERA2021167}, the penalty parameter $\lambda$ does not depend on $\ell$, and the LASSO penalty is scaled by $N_I(2\ell+1)$ rather than by $\sqrt{N_I(2\ell+1)}$. These modifications also affect the corresponding theoretical results on the concentration of the estimators around the true values, which are therefore included for completeness in Appendix \ref{sec:est_singleinterval}.
\end{remark}

Before stating the main result, we need to introduce an additional condition on the \emph{minimal signal-to-noise ratio}. Accordingly, we define the following quantity

$$
\brutti = \max\{48,32p\} \max\{C_\Phi, 1\}
 \max_{k=0,\dots,K} \sum_{\ell=0}^{\LN-1} \max \{q_\ell^{(k)},1\} \frac{\lambdaN^2}{\alpha_\ell},
$$
where $\alpha_{\ell} = \frac{1}{2} \frac{\minCZ}{\mu_{\max; \ell}}$. Note that, by assumption \ref{cond:ratio_cl}, $C_L$ is bounded when $L\to \infty$ under suitable conditions on the $\lambda_\ell$'s (see Theorem \ref{th:main_res} below).

\begin{condition}[Signal-to-noise ratio]\label{cond:snr} Let us define the minimal jump size $$\kappa_L= \underset{k =0, \dots, K}{\min}\sum_{\ell = 0}^{\LN-1}(2\ell+1) \| \bm{\phi}_\ell^{(k+1)} - \bm{\phi}_\ell^{(k)} \|^2_2,$$ that is the minimum distance between the autoregressive parameters of two consecutive segments. It holds the following bound
$$
\Delta \kappa_L\ge C_\epsilon(K\gamma+ (2K+1)\brutti).
$$
where $\gamma$ is the tuning parameters given in Definition \ref{def:pen_estim}, $\Delta >p\ge 1$ is the minimal spacing between consecutive change points, and $C_\epsilon>0$ is an absolute constant. 
\end{condition}

\begin{theorem}\label{th:main_res} Consider the spherical change point autoregressive process $\{T(x,t),\ x \in \mathbb{S}^2, \ t \in \mathbb{Z}\}$ defined as in Section \ref{sec:def_model}, such that Conditions \ref{cond:identifiability}, \ref{cond:causality}, \ref{cond:smoothness}, \ref{cond:boundedness}, \ref{cond:snr} hold. 
Let $$
 \FN= c_0 a_0 \,\maxCZ\left(1+ \frac{1}{\mu_{\min; \ell}} \right),
$$
for some  absolute constants $a_0,c_0 > 0$. 

Consider the change point estimators $\{\hat{\eta}_k\}_{k = 1}^{\hat{K}}$ obtained as a solution of the dynamic programming optimization problem given in \eqref{eq:min_problem}, \eqref{eq:loss} and \eqref{eq:phi_est} with tuning parameters 
$$
\gamma  \ge \brutti(10K+4)
$$
and
$$
\lambdaN\ge  4 \FN \sqrt{\log p\LN} \cdot \sqrt{\frac{32p^2 \sum_{\ell=0}^{\LN-1} (2\ell+1)\alpha_\ell}{\alpha_\ell}}, \quad \ell =0,\dots,\LN-1.
$$
Then, with probability at least $1 - c_1 e^{-c_2 \log(pL)}$,
$$
\hat{K} = K \quad \text{and} \quad |\hat{\eta}_k - \eta_k | \le C_\epsilon \left( \frac{K\gamma+ (2K+1)\brutti}{\kappa_\LN}\right)
$$
for some absolute constants $C_\epsilon,c_1, c_2 > 0 $.
\end{theorem}
\begin{proof}
It follows from Proposition \ref{pr:4cases} that, $K \le |\hat{\mathcal{P}}|\le3K$. This combined with Proposition \ref{pr:k=k} completes the proof.
\end{proof}

Theorem \ref{th:main_res} is based on Proposition \ref{pr:4cases} and \ref{pr:k=k} which will be stated and proved in Section \ref{sec:res_part}. Proposition \ref{pr:4cases} and \ref{pr:k=k} can be considered minimization problems working on two different levels: the lower level considers separately every fixed interval of the partition, analyzing how its observations behave, while the higher level considers the partition, i.e., the complete set of intervals, and its properties. 
Proposition \ref{pr:4cases} and \ref{pr:k=k} analyze in detail the higher level, considering the entire partition. Nevertheless, in order to prove these propositions, we require further results concerning the prediction error within a fixed interval. These results will be discussed in the following section, and their proofs rely on additional findings regarding spectral parameter estimation, as presented in Appendix \ref{sec:est_singleinterval}.

\section{Theoretical Results}\label{sec:theor_results}

Before introducing the theoretical results on the prediction error and the optimal partition, let us define a function that, given an interval of observations $I=[s,e]$ and a parameter vector $\bm{\beta}$, returns the reconstruction error obtained using $\bm{\beta}$ as the autoregressive parameter vector. This function will be handy in the statements and proofs of the following results.\\ 

For a fixed multipole $\ell$, we define the function $\mathcal{L}_{\ell}(\cdot)$ that, given the interval $I=[s,e]$ and a $p-$dimensional vector $\bm{\beta}$, returns the following reconstruction error
\begin{align*}
(I, \bm{\beta}) \mapsto\mathcal{L}_{\ell}(I, \bm{\beta}) = \sum_{t = s+p}^e \sum_{m=-\ell}^{\ell}\left| a_{\ell, m}(t) -  \bolda_{\ell,m}'(t-1){\bm{\beta}}\right|^2. 
\end{align*}

\subsection{On the Prediction Error}\label{sec:pred_err}

\begin{theorem}\label{th::error_diff}(Prediction Error) Consider the estimation problem \eqref{eq:phi_est} and assume Conditions \ref{cond:identifiability} and \ref{cond:causality} hold.
%Moreover suppose that
%$$
%\bigcap_{\ell = 0}^{\LN - 1}\widehat{\Gamma}_{\ell, I} \sim \operatorname{RE} \left( \alpha_{\ell}, \tau_{\ell} \right) \quad \text{a.s.}, \quad \text{with} \quad \qproxi_{\ell}\tau_{\ell}\le \alpha_{\ell} /32
%$$
%and the deviation condition is satisfied almost surely, that is,
%$$
%\bigcap_{\ell = 0}^{\LN-1}\left \{\left \| \frac{X'_{\ell, I}(\mathbf{Y}_{\ell, I}- X_{\ell, I}\phiproxi_{\ell})}{\sqrt{N(2\ell+1)}} \right\|_{\infty} \le  \FN\sqrt{\log(p\LN)} \right \} \quad \text{a.s..}
%$$ 
Consider the estimated autoregressive parameter vectors $\{\phihat_{\ell,I}\}_{\ell=0}^{\LN-1}$ solution of \eqref{eq:phi_est} and an interval $I$ that contains no true change point. Let $\{\phiproxi_{\ell}\}_{\ell=0}^{\LN-1}$ be the true parameter vectors in $I$ and $\{\qproxi_\ell\}_{\ell=0}^{\LN-1}$ are the associated sparsity indeces.

The prediction error over the complete set of coefficients $\{\phihat_{\ell,I}\}_{\ell=0}^{\LN-1}$ can be bounded by the quantity

\begin{align}
   &\left|\sum_{\ell = 0}^{\LN-1}\mathcal{L}_{\ell}(I, \phihat_{\ell,I}) -  \sum_{\ell = 0}^{\LN-1}\mathcal{L}_{\ell}(I, \phiproxi_{\ell})\right| \le  \max \{48, 32p\} \max\{C_\Phi, 1\}
\sum_{\ell=0}^{\LN-1} \max \{\qproxi_\ell,1\} \frac{\lambdaN^2}{\alpha_\ell}.\label{eq:prederr_phihat}
\end{align}
Moreover, for any given set of coefficients $\{\bm{\beta}_\ell\}_{\ell=0}^{\LN-1}$, where $\bm{\beta}_\ell \in \mathbb{R}^p,$ we can obtain the following upper bound 

\begin{align}
   &\sum_{\ell = 0}^{\LN-1}\mathcal{L}_{\ell}(I, \phiproxi_{\ell}) - \sum_{\ell = 0}^{\LN-1}\mathcal{L}_{\ell}(I, \bm{\beta}_\ell) \le  \max \{48, 32p\} \max\{C_\Phi, 1\}
\sum_{\ell=0}^{\LN-1} \max \{\qproxi_\ell,1\} \frac{\lambdaN^2}{\alpha_\ell}.\label{eq:prederr_beta}
\end{align}

\end{theorem}
The bounds in \eqref{eq:prederr_phihat} and \eqref{eq:prederr_beta} will be crucial in the proof of the main result. The Appendix \ref{sec:est_singleinterval} provides the additional lemmas necessary for demonstrating Theorem \ref{th::error_diff}.

\subsection{On the Optimal Partition}\label{sec:res_part}
The previous section contains the theoretical results for the estimation of the parameters in a fixed interval $I=[s,e]$. Now, we proceed by considering the entire partition. 

%For this reason, all the previously defined quantities will be indexed with the interval $I_j$ of interest. For example, the quantity $\phihat_{\ell, I}$ referred to the interval $I_j$ will be written as $\phihat_{\ell;I_j} $, the same will hold for all the other quantities defined throughout the paper. In addition, for any segment $I_j = [s,e]$, we define $|I_j|= N_{I_j} = e-s -p+1$.\\                                 

\begin{proposition}\label{pr:4cases} Under the same condition in Theorem \ref{th:main_res} and letting $\hat{\mathcal{P}}$ be the solution to \eqref{eq:min_problem}, with probability at least $1 - c_1 e^{-c_2 \log(pL)}$, for some absolute constants $c_1, c_2>0$, the following cases hold.\\

\noindent\textbf{Case 1.} Consider any $I = [s, e] \in \hat{\mathcal{P}}$ that contains one and only one true change point $\eta \in \{\eta_k, k=0,\dots,K\}$. Denote $I_1 = [s, \eta)$, $I_2 = [\eta, e]$. If it holds that 
\begin{equation}\label{eq:cond_lemma12}
    \sum_{\ell = 0}^{\LN-1}\mathcal{L}_{\ell}(I, \phihat_{\ell,I}) \le \sum_{\ell = 0}^{\LN-1}\mathcal{L}_{\ell}(I_1, \phihat_{\ell,I_1})  + \sum_{\ell = 0}^{\LN-1}\mathcal{L}_{\ell}(I_2, \phihat_{\ell,I_2})  + \gamma,
\end{equation}
it holds a.s. that
%$$\min \{|I_1|,|I_2|\} \le c_4\left( \frac{\gamma + 48 \lambdaN^2 \sum_{\ell=0}^{\LN-1} \frac{\sumq}{\alpha_{\ell}}}{\kappa_\LN}\right),$$
$$\min \{|I_1|,|I_2|\} \le C_\epsilon \left( \frac{\gamma+3\brutti}{\kappa_\LN}\right).$$
where $C_\epsilon>0$ is an absolute constant.\\

\noindent\textbf{Case 2.} Consider any $I = [s, e] \in \hat{\mathcal{P}}$ that contains exactly two true change points ${i_1},{i_2} \in \{\eta_k, k=0,\dots,K\}$, $i_1< i_2$. Denote $I_1 = [s, i_1)$, $I_2 = [{i_1}, i_2)$, $I_3 = [{i_2}, e]$. %Assume Conditions \ref{cond:identifiability} and \ref{cond:causality} hold, and the deviation condition and the compatibility condition (Propositions \ref{prop::dev_condition} and \ref{prop::re_compatibility}) are satisfied almost surely. 
If it holds that 
\begin{equation}\label{eq:cond_lemma13}
    \sum_{\ell = 0}^{\LN-1}\mathcal{L}_{\ell}(I, \phihat_{\ell,I}) \le \sum_{\ell = 0}^{\LN-1}\mathcal{L}_{\ell}(I_1, \phihat_{\ell,I_1})  + \sum_{\ell = 0}^{\LN-1}\mathcal{L}_{\ell}(I_2, \phihat_{\ell,I_2}) + \sum_{\ell = 0}^{\LN-1}\mathcal{L}_{\ell}(I_3, \phihat_{\ell,I_3}) + 2\gamma,
\end{equation}
then
$$
\max \{|I_1|,|I_3|\} \le C_\epsilon \left( \frac{2\gamma+ 5\brutti}{\kappa_\LN}\right)
$$
where $C_\epsilon>0$ is an absolute constant.\\

\noindent\textbf{Case 3.} Consider any $I = [s, e] \in \hat{\mathcal{P}}$ that contains no true change point. Then, it holds that
\begin{equation}\label{eq:th_lemma14}
    \sum_{\ell = 0}^{\LN-1}\mathcal{L}_{\ell}(I, \phihat_{\ell,I}) < \min_{s+p+1\le b \le e-p}\left \{\sum_{\ell = 0}^{\LN-1}\mathcal{L}_{\ell}([s,b), \phihat_{\ell,[s,b)})  + \sum_{\ell = 0}^{\LN-1}\mathcal{L}_{\ell}([b,e], \phihat_{\ell,[b,e]})\right\} + \gamma.
\end{equation}\\
    
\noindent\textbf{Case 4.} Consider any $I = [s,e]$ that contains exactly $J$ true change points $i_1,\dots,i_J \in \{\eta_k, k=0,\dots,K\}$, $J\ge 3$. It holds that
\begin{equation*}
    \sum_{\ell = 0}^{\LN-1}\mathcal{L}_{\ell}(I, \phihat_{\ell,I}) > \sum_{j=0}^{J}\sum_{\ell = 0}^{\LN-1}\mathcal{L}_{\ell}(I_j, \phihat_{\ell,I_j}) + J\gamma,
\end{equation*}
where $I_0 = [s, i_1)$, $I_j = [i_{j}, i_{j+1})$ for $2\le j \le J-1$, and $I_{J} = [i_{J}, e]$.\\
\end{proposition}

Loosely speaking, the four cases in Proposition \ref{pr:4cases} guarantee that every estimated interval of the partition, obtained as a solution of the minimization problem \eqref{eq:min_problem}, contains no more than two true change points and with good localization results. \\
Indeed, \textbf{Case 1} ensures that for each interval $\hat{I}= [s,e] \in \hat{\mathcal{P}}$ containing one and only one true change point $\eta$, the smallest interval is controlled, i.e., the distance between the true and the estimated change point is controlled by the quantity
$$\min \{e-\eta,\eta-s\} \le C_\epsilon \left( \frac{\gamma+ 3\brutti}{\kappa_\LN}\right).$$
\textbf{Case 2} ensures that for each interval $\hat{I}= [s,e] \in \hat{\mathcal{P}}$ containing exactly two true change point $\eta_1 <\eta_2$, the smallest interval is controlled, that is
$$\min \{e-\eta_2,\eta_1-s\} \le C_\epsilon \left( \frac{2\gamma+ 5 \brutti}{\kappa_\LN}\right).$$
\textbf{Case 3} ensures that all consecutive intervals $\hat{I}, \hat{J} \in \hat{\mathcal{P}}$, the interval $\hat{I} \cup \hat{J}$ contains at least one true change points and \textbf{Case 4} guarantees that no interval $\hat{I} \in \hat{\mathcal{P}}$ contains strictly more than two true change points.

\begin{proposition}\label{pr:k=k} Under the same condition in Theorem \ref{th:main_res} and letting $\hat{\mathcal{P}}$ be the solution to \eqref{eq:min_problem}, satisfying $K \le |\hat{\mathcal{P}}|\le3K$, it holds that $|\hat{\mathcal{P}}|=K+1$ with probability at least $1 - c_1 e^{-c_2 \log(pL)}$ for some absolute constants $c_1, c_2>0$.
\end{proposition}

\section{Numerical Results}\label{sec:num_res}

In this section, we assess the numerical performance of the proposed penalized estimator on both synthetic and real data. The change point detection capability is illustrated in a variety of controlled scenarios and on a dataset of global temperature anomalies. After presenting the numerical results for the synthetic settings, we provide insights regarding the role of the tuning parameters. \\ 

\textbf{Simulation settings.} Synthetic datasets were generated according to the change point SPHAR($p$) process described in Section~\ref{sec:def_model}. Performance is assessed using (i) the scaled Hausdorff distance and (ii) the mean estimated location of the change point(s). The scaled Hausdorff distance $\mathcal{D}$ is defined as
\begin{align*}
    \mathcal{D} \left( \left\{\hat{\eta}_k\right\}_{k=1}^{\hat{K}}, \left\{\eta_k\right\}_{k=1}^K\right) = \frac{d\left( \left\{\hat{\eta}_k\right\}_{k=1}^{\hat{K}}, \left\{\eta_k\right\}_{k=1}^K\right)}{n},
\end{align*}
where $d(\cdot,\cdot)$ denotes the Hausdorff distance between two compact sets in $\mathbb{R}$, given by
\begin{align*}
    d\left( A, B\right) = \max \left \{ \underset{a \in A}{\max}\, \underset{b \in B}{\min}\, |a-b|, \underset{b \in B}{\max}\, \underset{a \in A}{\min}\, |a-b|\right \}.
\end{align*}
Note that if $K,\hat{K} \ge 1$, then $\mathcal{D} \left( \left\{\hat{\eta}_k\right\}_{k=1}^{\hat{K}}, \left\{\eta_k\right\}_{k=1}^K\right) \le 1$. For convenience we set $\mathcal{D} \left( \emptyset, \left\{\eta_k\right\}_{k=1}^K\right)= 1$. \\

Given a set of observations for the timestamps $t \in \{1, \dots, n\}$, if $K$ change points $\eta_1, \dots, \eta_K$ belong to $\{1, \dots, n\}$, we can define the true location $ \rho_k= \eta_k /n \in [0,1]$ and its estimator $ \hat{\rho}_k= \hat{\eta}_k /n$ , for $k = 1, \dots, K$. The mean estimated location is measured by averaging the estimated location $ \hat{\rho}_k = \hat{\eta}_k / n$ over the simulations for each true change point. In the scenario with two change points $\eta_1 < \eta_2$ we follow the approach in \cite{Safikhani} where, for every simulation, each of the estimated change points is counted as a ``success'' for the first true change point $\eta_1$ if it is on the interval $[1, \eta_1 +0.5 (\eta_2-\eta_1))$, it is otherwise counted as a ``success'' for $\eta_2$ if it falls in the interval $[\eta_1 +0.5 (\eta_2-\eta_1),n]$.\\

We define the vector $\bm{\beta} = ( \beta_0, \dots, \beta_{\LN-1})'$, such that 
\begin{equation} \label{eq:coeff_sim}
    \beta_{\ell} = \begin{cases}
    0.9\times (\ell+1)^{-1/(8-d)} & \text{if } \ell < q \\
    0  & \text{if } \ell \ge q
\end{cases},
\end{equation}
where $q\le \LN-1$ is a parameter that defines the sparsity of $\bm{\beta}$, while the parameter $d<8$ determines the values of the different components, and it will define the separation between the parameter vectors of the different segments. \\

Note that, in all reported experiments, the tuning parameter $\lambda$ was held constant across all multipoles, i.e., $\lambda_\ell \equiv \lambda$ for every $\ell = 0, \ldots, L-1$. \\

\textit{Single change point scenario.} Samples of spherical random coefficients $a_{\ell,m}(t)$'s are generated with $n = 200$ and $\LN=10$, from a SPHAR(1) process, with a true change point at $\eta \in \{1, \dots, n\}$. For all timestamps $t= 1, \dots,  n$, the autoregressive coefficient $\phi_{\ell}(t)$ is defined as 
$$
\phi_\ell(t) = \begin{cases}
    {\phi}_{\ell}^{(0)} =- \beta_{\ell} & \text{if } t < \eta \\
     \phi_{\ell}^{(1)} = \beta_{\ell} & \text{if } t \ge \eta
\end{cases},
$$
for $\ell=0, \dots, \LN-1,$ where $\beta_{\ell}$ is defined as in \eqref{eq:coeff_sim}. In this setting, the parameter $q$ representing the number of non-zero components of the coefficient vectors is assumed to be the same for both segments. The parameter $d$ describes the distance between $\phi^{(0)}$ and $\phi^{(1)}$: increasing values of $d$ represent decreasing differences between the autoregressive parameters, in terms of Euclidean norm. The two autoregressive processes also differ in terms of noise variance. The power spectrum of the first segment is defined as $C^{(0)}_{0;Z}=1$ and $C^{(0)}_{\ell;Z} = \frac{1}{\ell (\ell+1)}$, for $\ell=0, \dots, \LN-1$, while for the second segment it is defined as $C^{(1)}_{0;Z}=0.5$ and $C^{(1)}_{\ell;Z} = \frac{0.5}{2\ell (\ell+1)},$ for $\ell=0, \dots, \LN-1$.\\

We consider different configurations of $q$ and $d$, for balanced (Scenario 1) and unbalanced (Scenario 2) sets of timestamps. 
Scenario 1 presents one single change point at $\eta = 100$, i.e., the true location is $\rho = \eta/n = 0.5$, meaning that we observe a balanced number of realizations from the two processes. Scenario 2 presents one single change point at $\eta = 50$, i.e., the true location is $\rho = \eta/n = 0.25$, meaning that we observe an unbalanced number of realizations from the two processes. For each setting, $M=100$ datasets are simulated. Results are reported in Table \ref{tab:singlecp}. \\ 
The method yields accurate localization and estimation of the change point across configurations. Detection is most challenging for high sparsity ($q=2$) and low signal separation. The setting with the best performance, both in the centered and non-centered scenarios, is the one with a small sparsity level and a large separation between the autoregressive parameters of the two segments, that is $q = 8$ and $d = 2$.\\

\begin{table}[ht!]
\centering
\begin{tabular}{c|cc|cc}
\hline\hline
\multirow{2}{*}{Setting} & \multicolumn{2}{c|}{Scenario 1: $\rho=0.50$} & \multicolumn{2}{c}{Scenario 2: $\rho=0.25$} \\
 & $\mathcal{D}$  & mean($\hat{\rho}$) & $\mathcal{D}$  & mean($\hat{\rho}$) \\
\hline
$q = 8,\ d = 4$ & 0.0018 (0.0024) & 0.498 (0.002) & 0.0020 (0.0024) & 0.248 (0.002) \\
$q = 8,\ d = 2$ & 0.0010 (0.0020) & 0.499 (0.002) & 0.0010 (0.0020) & 0.249 (0.002) \\
$q = 2,\ d = 4$ & 0.0026 (0.0029) & 0.498 (0.003) & 0.0026 (0.0034) & 0.249 (0.004) \\
$q = 2,\ d = 2$ & 0.0026 (0.0025) & 0.498 (0.003) & 0.0025 (0.0029) & 0.248 (0.003) \\
\hline
\end{tabular}
\caption{\textbf{Performance of the penalized estimator for single change point scenarios.} Scenario 1: balanced ($\eta=100$), Scenario 2: unbalanced ($\eta=50$). Standard errors are reported in parentheses.}
\label{tab:singlecp}
\end{table}

\textit{Multiple change points scenario.} We also consider an epidemic change setting, with $n=225$ and two true change points at $\eta_1=75$ and $\eta_2=150$. The model parameters revert after the second change point to the initial values. Specifically, for all timestamps $t= 1, \dots,  n$, the autoregressive coefficient $\phi_\ell(t)$ is defined as  
$$
\phi_{\ell}(t) = \begin{cases}
    \phi_{\ell}^{(0)} =- \beta_{\ell} & \text{if } t < \eta_1 \\
     \phi_{\ell}^{(1)} = \beta_{\ell} & \text{if } \eta_1 \le t < \eta_2\\
     \phi_{\ell}^{(2)} = \phi_{\ell}^{(0)} =- \beta_{\ell} & \text{if } t \ge \eta_2,
\end{cases}
$$
for $\ell=0, \dots, \LN-1$. Table \ref{tab:2cp} reports the results for different settings of $(q,d)$ over $M=100$ replications.

\begin{table}[ht!]
    \centering
    \begin{tabular}{cccc}
    \hline \hline
        Setting  & $\mathcal{D}$ & mean($\hat{\rho}$) \\
        \hline
        \multirow{2}{*}{$q = 8,\ d = 4$ }& 0.3351 & [0.332, 0.667] \\
         & \small{(0.0022)} & \small{(0.002, 0.003 )}\\
        \multirow{2}{*}{$q = 8,\ d = 2$} & 0.3343  & [0.333,  0.666] \\
         & \small{(0.0018)} &  \small{(0.002, 0.002) }\\     
        \multirow{2}{*}{$q = 2,\ d = 4$} & 0.3369 & [0.332,  0.669] \\
         & \small{(0.0039)} & \small{(0.004, 0.004)} \\
        \multirow{2}{*}{$q = 2,\ d = 2$} & 0.3361 & [0.332, 0.668] \\
         & \small{(0.003)} & \small{(0.003, 0.003)} \\
        \hline
    \end{tabular}
    \caption{\textbf{Performance of the penalized estimator for two change points scenarios.} Two change points at $\eta_1=75$, $\eta_2=150$. Standard errors in parentheses.}
    \label{tab:2cp}
\end{table}

Although the Hausdorff distance is larger than the one observed in the single change point scenarios due to the increased complexity, the estimator still successfully identifies both change point positions with low error. \\

\textbf{Tuning parameters role.} In Theorem \ref{th:main_res} we defined the optimal theoretical value of the parameters $\lambda$ and $\gamma$. Here we show how the parameters choice impacts the detection results. With this goal in mind, we select a simulation setting and test different parameter configurations. Specifically, we consider the setting $(q=8, d=2)$ with two true change points $\eta_1=75, \eta_2 = 150$, as in Scenario 3. Figure \ref{fig:sett_alpha_gamma} shows the distribution of the estimated change points for six different configurations corresponding to $\lambda \in \{0, 1\}$ and $\gamma \in \{100, 200, 300\}$.
\begin{figure}[ht!]
\centering
\subfloat{
\includegraphics[width=0.45\textwidth]{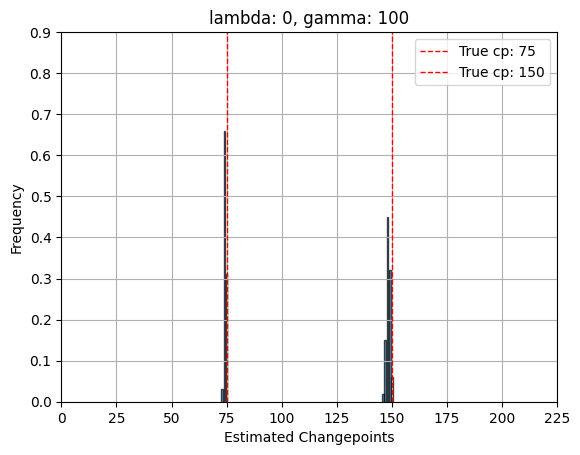}}
\qquad \subfloat{
\includegraphics[width=0.45\textwidth]{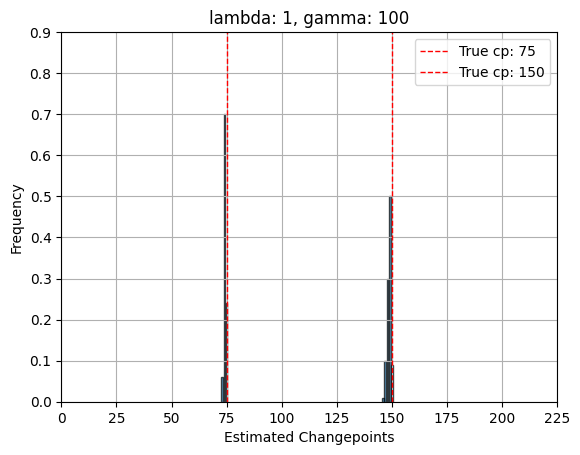}}
\\
\centering
\subfloat{
\includegraphics[width=0.45\textwidth]{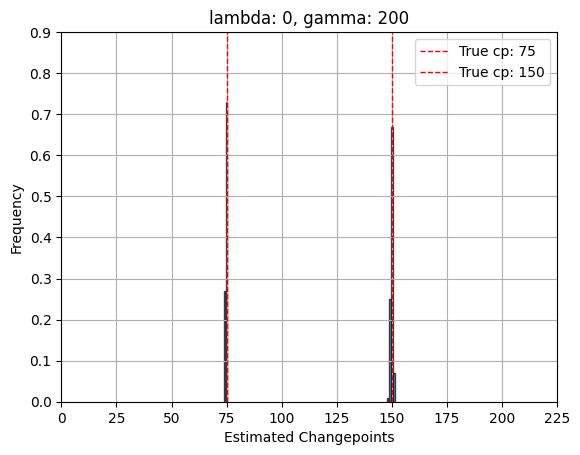}}
\qquad \subfloat{
\includegraphics[width=0.45\textwidth]{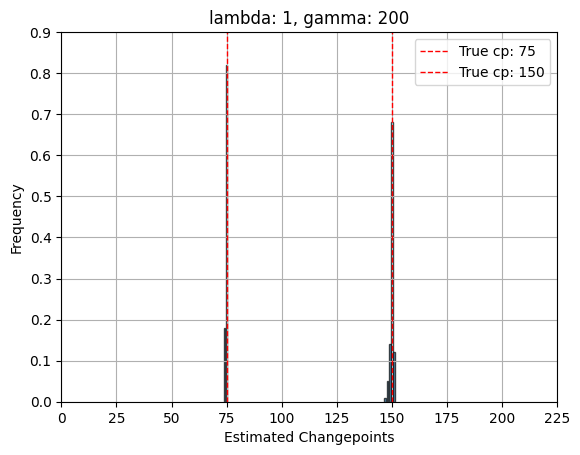}}
\\
\centering
\subfloat{
\includegraphics[width=0.45\textwidth]{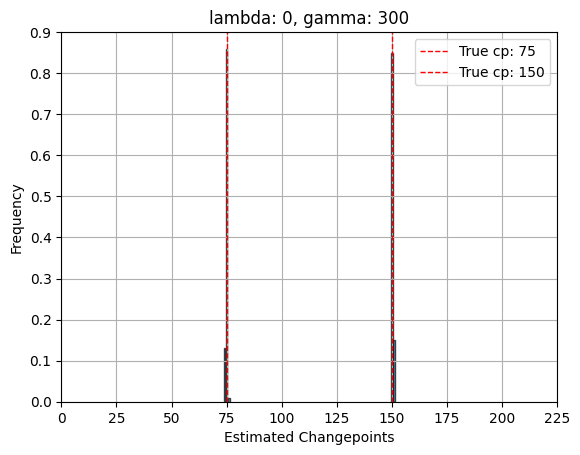}}
\qquad \subfloat{
\includegraphics[width=0.45\textwidth]{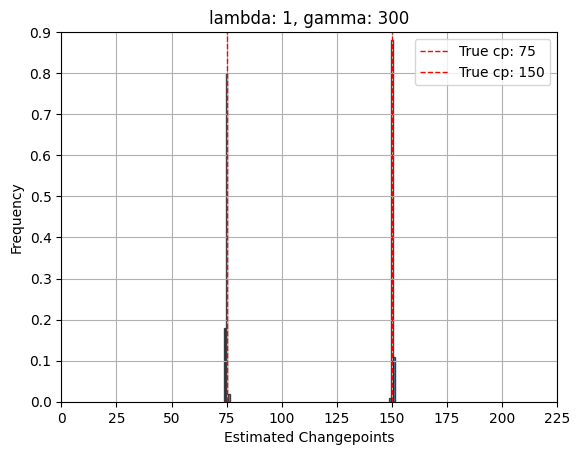}}
\caption{Estimated change point distribution for different values of $(\lambda,\ \gamma)$, $M=200$ simulations per setting.}
\label{fig:sett_alpha_gamma}
\end{figure}

From the figure, it is evident the crucial role of the parameter $\gamma$ and how it is able to penalize the number of estimated change points and center their estimation: for increasing values of $\gamma$, more and more simulations prefer a ``conservative'' solution, being more parsimonious in the choice of the change points and waiting for a longer period before assigning the change point, which improves the detection. The impact of the parameter $\lambda$ appears less significant due to the dense data configuration.Table \ref{tab:my_label} presents the Hausdorff distance of the analyzed settings.

\begin{table}[ht!]
    \centering
    \begin{tabular}{ccc}
    \hline \hline
         & $\lambda = 0$ & $\lambda = 1$\\
         \hline
        \multirow{2}{*}{$\gamma = 100$} & 0.3365 & 0.337\\
        & \small{(0.0023)} & \small{(0.0023)}\\
        \hline
        \multirow{2}{*}{$\gamma = 200$} & 0.3348  & 0.3346 \\
        & \small{(0.0021)} & \small{(0.002)} \\
        \hline
        \multirow{2}{*}{$\gamma = 300$} & 0.3345 & 0.3346\\
        & \small{(0.0019)} & \small{(0.002)}\\
        \hline
    \end{tabular}
    \caption{Hausdorff distance (sd) for different configurations $(\lambda, \gamma)$ analyzed over the data configuration $(q=8, d=2)$ with two true change points $\eta_1=75, \eta_2 = 150$. $M=100$ simulations per each setting.}
    \label{tab:my_label}
\end{table}

After evaluating all the configurations, the combination $(\lambda = 0,\ \gamma = 300)$ turned out to be the best set. This is because the high penalization results in a more accurate detection since the algorithm takes more time (and more observed realizations) before the identification of the change point. \\

Another hyperparameter influencing the detection results is $\delta$, representing the position of the first admissible change point candidate for both methods, or equivalently, the length of the first segment before a potential split is considered by our algorithm. This parameter can be thought as the empirical counterpart of $\Delta$ in \eqref{eq:delta} and in the evaluated settings has default value equal to $5$. Opting for a low value of $\delta$ may introduce instability during the initial loss comparison, potentially leading to the erroneous detection of a change point at the outset of the observations. \\

\textbf{Real data application.} The proposed method was applied to annual global (land and ocean) surface temperature anomalies from 1948 to 2020, constructed from the NCEP/NCAR monthly averages \cite{ncep}. Following the World Meteorological Organization policy, temperature anomalies are obtained by subtracting the long-term monthly means relative to the 1981--2010 base period. They are then averaged over months to switch from a monthly scale to an annual scale.

By means of the \texttt{healpix} package (see \cite{Gorski_2005} and the official \texttt{healpix} \href{https://healpix.sourceforge.io/}{\underline{website}}), we converted the gridded data into spherical maps with a resolution of $12\cdot\operatorname{NSIDE}^2$ pixels (NSIDE = 16) and then we computed the Fourier coefficients up to $\LN = 2\cdot\operatorname{NSIDE}$. We used the detection technique introduced in this dissertation to identify, if present, one or more changes in the temperature anomalies.

In order to handle possible anisotropies in the mean, for each segment $k=0,1$, we introduced an intercept $\mu_k \in L^{2}(\mathbb{S}^2)$, which has a representation in terms of spherical harmonics
$$
\mu_{k} = \sum_{\ell=0}^\infty \sum_{m=-\ell}^\ell \mu_{\ell,m;k} Y_{\ell,m}, \qquad \text{in }
 L^2(\mathbb{S}^2),$$
 with $\mu_{\ell,m;k} = \langle \mu_k, Y_{\ell,m} \rangle_{L^2}$.\\

Note that, for the real data application as well, the tuning parameter $\lambda$ was set to the same value for all multipoles, that is, $\lambda_\ell \equiv \lambda$ for every $\ell = 0, \ldots, L-1$. With a choice of hyperparameters $\lambda = 1$ and $\gamma = 20$, the method detected two change points at $\hat{\eta}_1 = 1976$ and $\hat{\eta}_2=1998$. After the detection, we can estimate the functional parameters $(\mu_k, \Phi_k)$ by solving the following least-squares minimization problem, see \cite{tesi_ale,caponera2019asymptotics}, 
\begin{align*}
    (\hat{\mu}_k,\, \hat{\Phi}_k) & := \mathrm{argmin} \,\sum_{t \in \mathscr{T}_k} \left \| T_t - \mu_{k;L} - \sum_{j = 1}^p \Phi_{k,j;L}T_{t-j} \right \|_{L^2(\mathbb{S}^2)}^2,
\end{align*}
where $\mu_{k;L}$ and $\Phi_{k,j;L}$ are the truncated version of $\mu_k$ and $\Phi_{k,j}$, respectively, for $j = 1, \dots, p,\ k=0, 1, 2$, and $\mathscr{T}_j$ is the set of timestamps belonging to each segment, i.e., $\mathscr{T}_0=\{1949, \ldots, \hat{\eta}_1-1\}$, $\mathscr{T}_1=\{\hat{\eta}_1, \ldots, \hat{\eta}_2 -1\}$, and $\mathscr{T}_2=\{\hat{\eta}_2, \ldots, 2020\}$.

The comparison between the three periods can be carried out by computing the three mean surfaces
$$
(I_L-\hat{\Phi}_k)^{-1} \hat{\mu}_k = \sum_{\ell=0}^L \sum_{m=-\ell}^\ell \frac{\hat{\mu}_{\ell,m;k}}{1-\hat{\phi}^{(k)}_{\ell;1k} - \dots - \hat{\phi}^{(k)}_{\ell;p}} Y_{\ell,m}, \qquad k=0,1,2,
$$
shown in Figure \ref{fig:maps2} together with their sign. The mean surface before the year 1976 has a prevalence of negative signs; then the detection technique identifies a more uniform period between the year 1976 and the year 1982. Finally, the average surface anomalies after the year 1998 have a prevalence of positive signs, with a strong increase in value at the poles.\\

\begin{figure}[ht!]
\centering
\subfloat{
\includegraphics[width=0.4\textwidth]{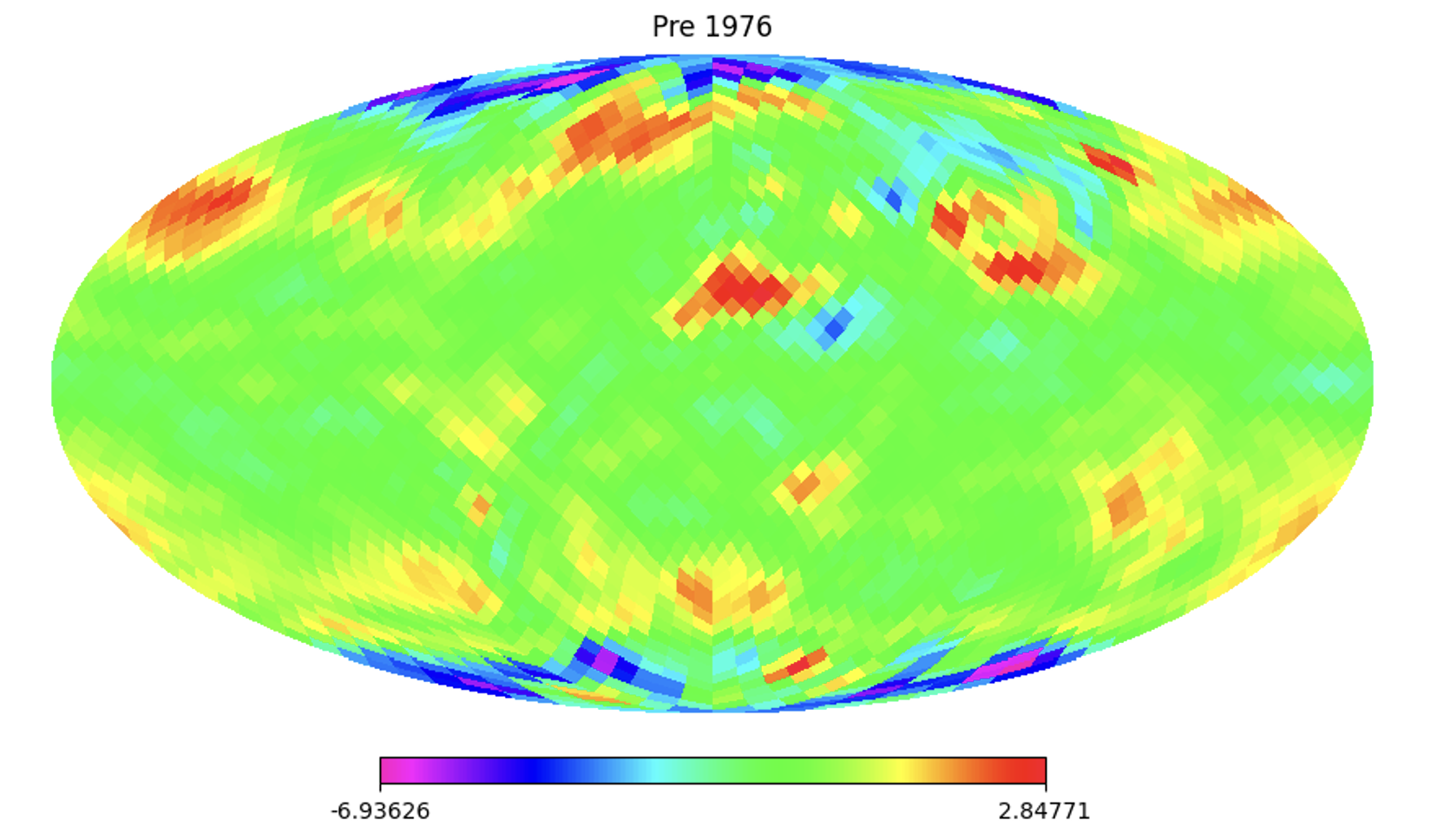}}
\qquad \subfloat{
\includegraphics[width=0.4\textwidth]{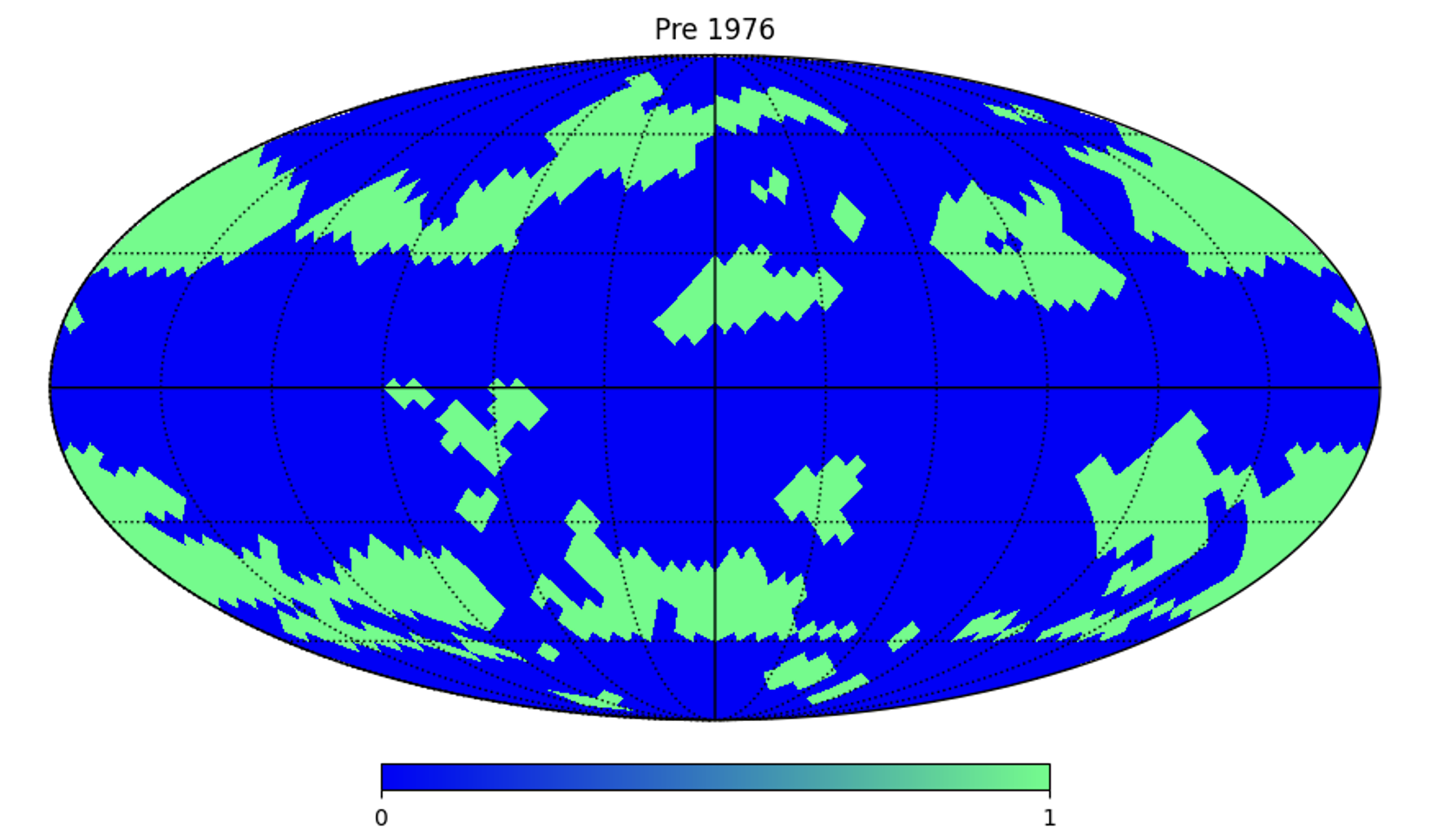}}
\\
\centering
\subfloat{
\includegraphics[width=0.4\textwidth]{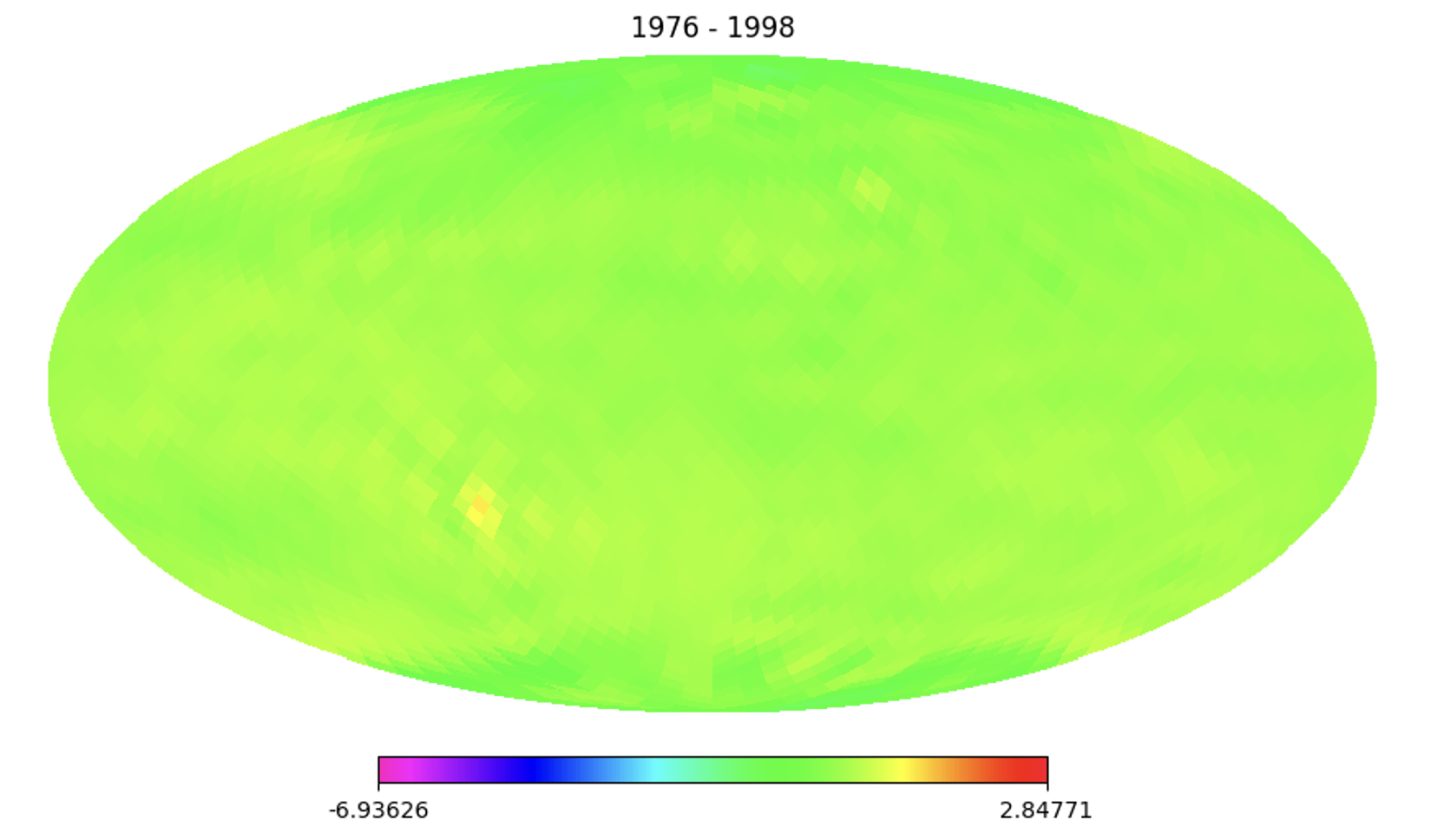}}
\qquad \subfloat{
\includegraphics[width=0.4\textwidth]{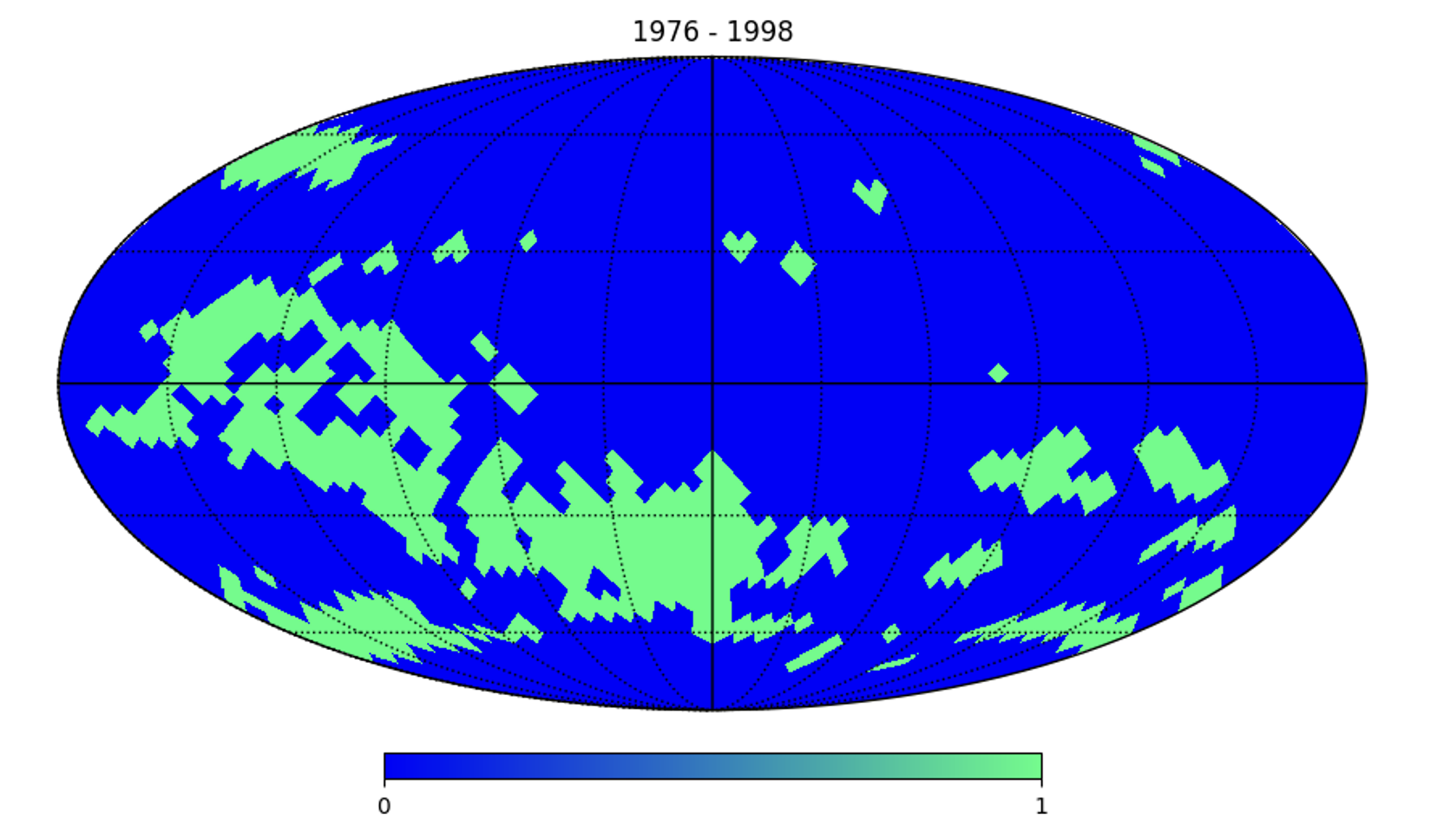}}
\\
\centering
\subfloat{
\includegraphics[width=0.4\textwidth]{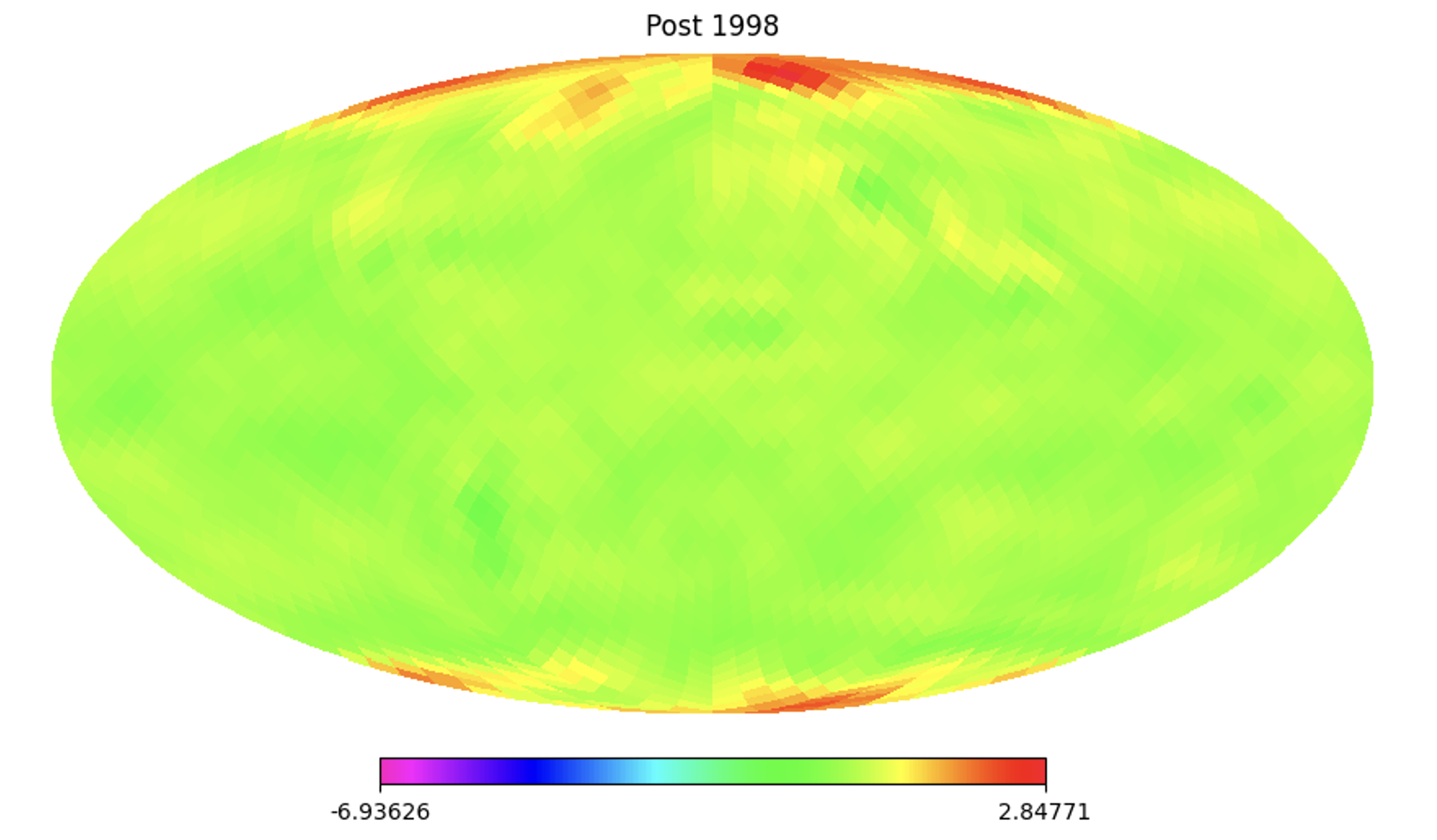}}
\qquad \subfloat{
\includegraphics[width=0.4\textwidth]{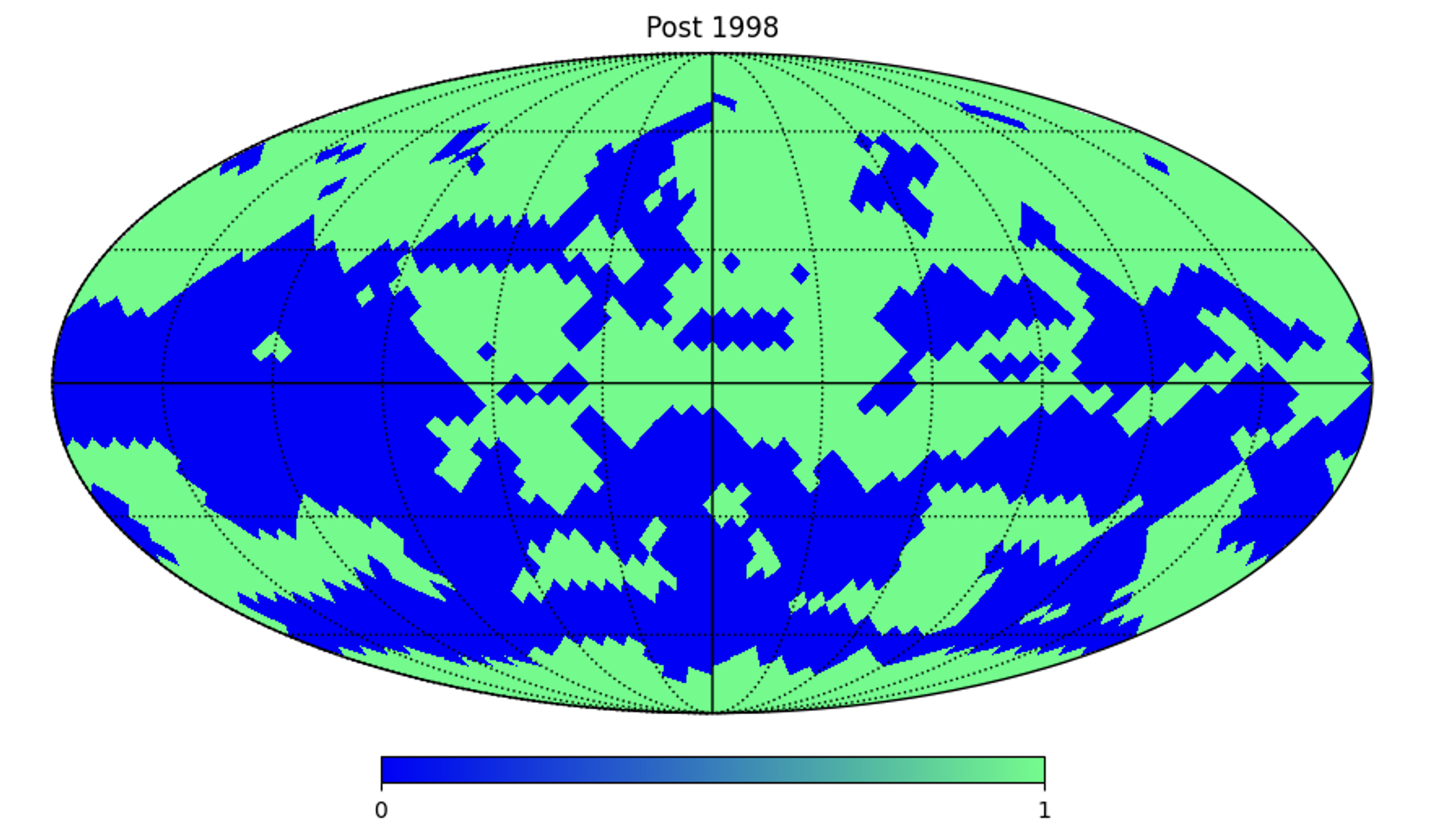}}
\caption{\small{\textbf{LASSO penalized estimator.} Left: estimated mean surfaces pre $\hat{\eta}_1 = 1976$, intermediate period $[\hat{\eta}_1,\hat{\eta}_2)=[1976, 1998)$ and post $\hat{\eta}_2=1998$ (on the same color scale). Right: corresponding negative (0) and positive (1) pixels.}}
\label{fig:maps2}
\end{figure}

A transition from predominantly negative to predominantly positive anomalies is observed, particularly at the poles, in agreement with patterns of global warming. A more detailed climatological investigation is warranted to relate these detected structural changes to specific climate regimes or events.

\section{Conclusions}\label{sec:concl}

This paper introduces the first general framework for change point detection in spherical functional autoregressive processes. By adapting penalized dynamic programming to the infinite-dimensional setting of SPHAR($p$) models, we provide a principled and theoretically justified approach for detecting structural breaks in spatio-temporal random fields defined on the sphere.

While extensive research exists on change point detection for scalar and vector-valued time series--including dependent and high-dimensional cases--prior literature is mainly focused on Euclidean data or, at best, on non-functional multivariate settings. The spherical functional time series setting introduces additional challenges, such as the infinite-dimensionality of the coefficient space and the interplay between spatial and temporal dependence. In this setting, we introduce a novel Spherical Autoregressive Change Point framework, which extends SPHAR($p$) models by allowing for structural breaks in the autoregressive dynamics governing the field.

Our main methodological contribution is the development and analysis of a penalized change point estimator, formulated via a LASSO-regularized dynamic programming approach. This estimator does not require any prior knowledge on the number or location of change points and is accompanied by finite-sample probabilistic guarantees for both the number and locations of the estimated change points and for segment-wise spectral parameter estimation. Theoretical results were established for the spectral parameter estimation error and for the accuracy of the estimated segmentation, with finite-sample rates explicitly depending on the process dimension and tuning parameters. Importantly, our framework accommodates the infinite-dimensional setting typical of functional data on the sphere by introducing a maximal observed multipole $\LN$, and the model and theory are robust as $\LN$ increases.

Our analysis also underscores how the maximum observed multipole $\LN$ enters both practical estimation and theoretical rates, and suggests further investigation into optimal choices of $\LN$ as a function of sample size or field resolution. Beyond this, the framework admits natural extensions, including (i) modeling change points affecting only specific multipoles or frequency bands, (ii) relaxing the independence assumption between segments, and (iii) devising formal inference procedures or tests for multiple change point scenarios.

By integrating and extending methodologies from change point analysis and spherical functional autoregressions, we overcome both the infinite-dimensionality of the functional setting and the specific challenges associated with concatenated SPHAR($p$) models under non-stationarity. The proofs and theoretical developments carefully account for the misspecification and boundary effects introduced by structural breaks, providing error bounds that reflect the specifics of the concatenated (block-dependent) model structure.

Application of the method to global temperature anomaly data illustrates its potential for uncovering nonstationarities of practical and scientific relevance--identifying not only temporal change points but also their spatial (harmonic) signatures. The ability to jointly detect, localize, and attribute regime shifts in the spatio-temporal dynamics of spherical fields opens the door to more nuanced investigations in climate science, planetary data analysis, cosmology, and beyond, where changes are often global yet spatially structured. Importantly, the harmonic domain approach enables practitioners to investigate whether changes are driven by low- or high-frequency spatial components, thus supporting more interpretable science.

Some practical limitations and implementation decisions are worthy of further discussion. First, the proposed procedure relies on the expansion of the random field in spherical harmonics up to a maximal multipole $L$, which in practice is dictated by the spatial resolution or the data's signal-to-noise ratio. The detection performance, as both theory and simulation confirm, can be sensitive to the choice of $L$ and the sparsity of the underlying model. Second, tuning parameter selection (for both penalization and minimal segment length) remains, as in the classical case, a challenging issue; while guidance from theory exists, adaptive or data-driven approaches--for example, using cross-validation, information criteria, or resampling--may further improve empirical performance. Finally, while our theoretical developments assume independence between stationary segments, real data may exhibit dependence across change points, which could affect finite-sample performance.

Several natural extensions arise from this work. One direction is the development of formal hypothesis tests and model selection procedures tailored to the spherical functional setting, perhaps using recent advances in selective inference or subsampling. Another is to investigate models where change points affect only a subset of harmonic coefficients or multipole bands, corresponding, for example, to spatially localized regime shifts or frequency-specific changes. Robust inference and uncertainty quantification, as well as relaxing assumptions on independence or stationarity within segments, also merit deeper exploration. Beyond SPHAR($p$) models, adapting these ideas to more general classes of spatial and spatio-temporal processes on manifolds remains an open and exciting area. 

Overall, the penalized estimator presented here offers a new and flexible tool for the detection of structural breaks in spherical random fields. It readily admits adaptation to applications in climate science, atmospheric data, cosmology, and more, and opens several avenues for future methodological and applied research.

\bibliographystyle{ieeetr}
\bibliography{biblio}

\newpage
%%%%%%%%%%%%%%%%%%%%%%%%%%%%%
\appendix
\section{Auxiliary Results}\label{sec:est_singleinterval}

This appendix presents auxiliary results on the spectral parameters estimation in a fixed interval $I = [s,e]$. Specifically, we derive classical deviation bounds and oracle inequalities for the $\widehat{\phi}_{\ell, I}$'s defined in \eqref{eq:phi_est}.

\medskip
In order to simplify the notation, we use $N$ to indicate $N_I = e - s - p + 1$ and we define, for for $\ell = 0, \dots, \LN-1$, the following $N(2\ell + 1)-$dimensional vectors
\begin{align*}
    & \mathbf{Y}_{\ell, I} = (a_{\ell, -\ell}(e), \dots,a_{\ell, -\ell}(s+p), \dots, a_{\ell, \ell}(s+p))', \\
    & \mathbf{Y}_{\ell, I}(j) = (a_{\ell, -\ell}(e-j), \dots,a_{\ell, -\ell}(s+p-j), \dots, a_{\ell, \ell}(s+p-j))', \quad j=1, \dots, p,\\
    & \mathbf{E}_{\ell, I} = (a_{\ell, -\ell; Z}(e), \dots,a_{\ell, -\ell; Z}(s+p), \dots, a_{\ell, \ell; Z}(s+p))'.
\end{align*}
We can thus define the $N (2\ell + 1) \times p$ matrix
$$
X_{\ell, I} = \{\mathbf{Y}_{\ell, I}(1): \cdots: \mathbf{Y}_{\ell, I}(p)\},
$$
and the empirical covariance matrix
$$
\widehat{\Gamma}_{\ell, I} = \frac{X'_{\ell, I}X_{\ell, I}}{N(2\ell+1)}.
$$

Note that if the interval $I$ does \emph{not} contain any true change points, then it must be contained within a single stationary segment, i.e., $I \subseteq [\eta_k, \eta_{k+1})$ for some $k = 0, \dots, K$. In this case, the model on $I$ is governed by a single parameter vector, and we have
\[
\mathbf{Y}_{\ell, I} = X_{\ell, I}\bm{\phi}^{(k)}_{\ell} + \mathbf{E}_{\ell, I},
\]
which implies
\begin{equation}\label{eq:emp_proc_nocp}
\frac{X'_{\ell, I}\mathbf{E}_{\ell, I}}{N(2\ell+1)} 
= \frac{X'_{\ell, I}\bigl(\mathbf{Y}_{\ell, I} - X_{\ell, I}\bm{\phi}^{(k)}_{\ell}\bigr)}{N(2\ell+1)}.
\end{equation}
Conversely, if $I$ contains at least one change point, this representation no longer holds.

\subsection{Deviation Bounds}
Following \cite{CAPONERA2021167}, we now apply the concept of the stability measure, introduced in Section \ref{sec:stability_meas}, to establish deviation bounds that characterize the concentration of sample covariances and the empirical processes around their expected values.

While \cite{CAPONERA2021167} focuses exclusively on the setting without change points, here we consider a generic interval
$I$, and the deviation bounds in Proposition \ref{prop::dev_bound} remain valid irrespective of the presence of change points. In the specific case where there are no change points, \eqref{eq:devbound4} can be related to \eqref{eq:emp_proc_nocp}. However, even in this case, the result differs slightly from \cite[Equation (4.6)]{CAPONERA2021167}, as we employ a distinct and more refined argument in its proof. This approach naturally extends the results to centered products of vector-valued random variables, provided that one can bound the largest eigenvalue of their corresponding covariance matrices from above.

As a complementary step, we first establish an additional lemma on the expectation of $\widehat{\Gamma}_{\ell, I}$, which we denote by $\Gammastar$.
\begin{lemma}\label{lemma:gamma}
Consider $\ell \in \mathbb{N}_0$ and the matrix $\Gammastar = \mathbb{E}\left[\widehat{\Gamma}_{\ell, I} \right].
$ Then,
    $$
2\pi \mathfrak{m}_\ell \le \lambda_{\min} (\Gammastar) \le \lambda_{\max} (\Gammastar) \le 2\pi \mathcal{M}_\ell.
$$
\end{lemma}

\begin{proposition}[Deviation Bounds]\label{prop::dev_bound}
Assume that Conditions \ref{cond:identifiability} and \ref{cond:causality} hold. Then, there exists
a constant $c > 0$ such that for any $\ell \in \mathbb{N}_0$, any $r-$sparse vectors $u,v \in \mathbb{R}^p$ with $\|u\|, \, \|v\| \le 1, \, r\ge 1$,
and any $\eta \ge0$, it holds that
\begin{align}
    &P\left(\left|v'(\widehat{\Gamma}_{\ell, I}-\Gammastar)v\right|>2 \pi r \mathcal{M}_\ell \eta\right) \le 2 \exp\left[-c \,N\, (2\ell+1) \,\min\{\eta^2, \eta\}\right], \label{eq:devbound1} \\
    &P\left(\left|u'(\widehat{\Gamma}_{\ell, I}-\Gammastar)v\right|>12 \pi r\mathcal{M}_\ell \eta\right) \le 6 \exp\left[-c \,N\, (2\ell+1) \,\min\{\eta^2, \eta\}\right] \label{eq:devbound2}
\end{align}    
    where $N = e-s-p+1$. In particular, for any $i, j \in \{1, \dots, p\}$, it holds that
\begin{align}  
    &P\left(\left|\left(\widehat{\Gamma}_{\ell, I}-\Gammastar\right)_{i,j}\right|>12 \pi \mathcal{M}_\ell \eta\right) \le 6 \exp\left[-c \,N\, (2\ell+1) \,\min\{\eta^2, \eta\}\right]. \label{eq:devbound3}
\end{align}
Moreover,
\begin{align}  
&P\left(\left\|\frac{X'_{\ell, I} \mathbf{E}_{\ell, I}}{N(2\ell+1)}\right\|_{\infty}> a_0 \maxCZ\left(1+ \frac{1}{\mu_{\min; \ell}} \right)
    \eta\right) \notag\\&\le 6p \exp\left[-c \,N\, (2\ell+1) \,\min\{\eta^2, \eta\}\right],
    \label{eq:devbound4}
\end{align}
with $a_0=3/2$.
\end{proposition}

Specifically, \eqref{eq:devbound1} will be used to verify the \emph{compatibility condition} (see Proposition \ref{prop::re_compatibility}), while \eqref{eq:devbound4} will be used to prove the \emph{deviation condition} (see Proposition \ref{prop::dev_condition}).

\begin{remark}
Note that we are considering a generic interval 
$I$, and regardless of the presence of change points, the deviation bounds in Proposition \ref{prop::dev_bound} hold. However, in the specific case where no change points are present, \eqref{eq:devbound4} can be linked to \eqref{eq:emp_proc_nocp}. It also slightly differs from \cite[Equation (4.6)]{CAPONERA2021167}, as we employ a different and improved argument to prove it. This approach readily extends the results to centered products of vector-valued random variables, provided that one can bound from above the largest eigenvalue of their corresponding covariance matrices.
\end{remark}

\subsection{Oracle Inequalities Derivation}\label{sec:oracle_deriv}
Once the bounds for the covariances are established, we proceed to present the classical steps necessary for the LASSO in our setting. We start with the so-called \emph{basic inequality}, an essential result based solely on the definition of the LASSO estimator, assuming a linear underlying model. We then consider the \emph{deviation condition} and the \emph{compatibility condition}, and conclude with the \emph{oracle inequalities}, which contribute to obtaining the upper bound on the prediction error given in Theorem \ref{th::error_diff}.

In this section, unless stated otherwise, we focus on an interval \(I\) without change points. For \(\ell = 0, \dots, \LN-1\), we let \(\phiproxi_\ell\) denote the corresponding true spectral parameter, i.e., the one satisfying
\[
\mathbf{Y}_{\ell, I} = X_{\ell, I}\phiproxi_\ell + \mathbf{E}_{\ell, I},
\]
and we denote by \(\qproxi_{\ell}\) the number of non-zero entries of \(\phiproxi_{\ell}\).

However, even within an interval \(I\) that contains no change points, the results differ from \cite{CAPONERA2021167}, since we introduce a LASSO penalty parameter that may depend on the multipole \(\ell\) and is further scaled by \(\sqrt{N(2\ell+2)}\) rather than \(N(2\ell+1)\) (see \cite[Equation (3.1)]{CAPONERA2021167}).
Together with some additional adjustments, this modification contributes to an improvement of the result by a factor of \(\sqrt{2\ell+1}\) in the oracle inequalities.

\begin{proposition}[Basic Inequality]\label{prop::basic_inequality}
    Consider the estimation problem in \eqref{eq:phi_est}. For any $\ell \in \{ 0, \dots, \LN-1\}$, set $\bm{\Delta}_\ell = \widehat{\bm{\phi}}_{\ell, I} - \phiproxi_{\ell}$. Then, the following basic inequality holds 
\begin{equation}
\bm{\Delta}'_{\ell}\widehat{\Gamma}_{\ell, I}\bm{\Delta}_\ell 
\le \frac{2\bm{\Delta}'_{\ell}X'_{\ell, I}(\mathbf{Y}_{\ell, I}- X_{\ell, I}\phiproxi_{\ell})}{N(2\ell+1)} 
+ \frac{\lambda_\ell}{\sqrt{N(2\ell+1)}} \left( \|\phiproxi_{\ell}\|_1 - \|\phiproxi_{\ell}+ \bm{\Delta}_{\ell}\|_1 \right).
\label{eq:basic_ineq}
\end{equation}
\end{proposition}

The basic inequality shows that the prediction error 
\(\bm{\Delta}'_{\ell}\widehat{\Gamma}_{\ell, I}\bm{\Delta}_\ell\) 
is bounded by the sum of two terms: the first is a random term, while the second is deterministic and depends on the penalty parameter \(\lambda_\ell\).

To handle the stochastic part of the basic inequality, we introduce an event \(\mathscr{S}_L\) under which, for each \(\ell = 0, \dots, \LN - 1\), the fluctuations of the random component
\[
\frac{X'_{\ell, I}\bigl(\mathbf{Y}_{\ell, I} - X_{\ell, I}\phiproxi_{\ell}\bigr)}{N(2\ell+1)}
\]
are controlled by a deterministic bound that may depend on \(\ell\), with high probability. This non-uniformity is a direct consequence of the fact that the penalty parameter \(\lambda\) depends on \(\ell\); as a result, the deterministic bounds controlling the stochastic fluctuations are also \(\ell\)-dependent. Unlike in \cite{CAPONERA2021167}, the bounds are therefore not uniform over \(\ell\). We define the event \(\mathscr{S}_L\) as follows.

\begin{definition}
    In the setting previously described, let
    $$
    \mathscr{S}_L = \bigcap_{\ell = 0}^{\LN-1}\left \{\left \| \frac{X'_{\ell, I}(\mathbf{Y}_{\ell, I}- X_{\ell, I}\phiproxi_{\ell})}{\sqrt{N(2\ell+1)}} \right\|_{\infty} \le  \FN\sqrt{\log(p\LN)}\right\}
    $$
    where $\FN$ is a deterministic function depending only on the autoregressive parameters, through $\mu_{\min; \ell}$, and noise variances, through $\maxCZ$. The deviation condition is said to hold if the event $\mathscr{S}_L$ happens.
\end{definition}

The following theorem shows that, for an appropriate choice of $\FN$ and $N$, the event $\mathscr{S}_L$ has high probability to occur.

\begin{proposition}[Deviation Condition]\label{prop::dev_condition}
Consider the estimation problem in \eqref{eq:phi_est} and assume that Conditions \ref{cond:identifiability} and \ref{cond:causality} hold. There exist some constants $c_0,c_1,c_2>0$ such that, if we define
%$$
% \FN= c_0 \, \underset{\ell <\LN}{\max} \left[a_0 \maxCZ\left(1+ \frac{3 + (1+M)pC_\Phi}{\mu_{\min; \ell}} \right)\right],
%$$
$$
 \FN= c_0 a_0 \,\maxCZ\left(1+ \frac{1}{\mu_{\min; \ell}} \right),
$$
and if $N \succeq \log (p \LN)$, then
\begin{align}
    P\left( \bigcap_{\ell = 0}^{\LN-1}\left \| \frac{X'_{\ell, I}(\mathbf{Y}_{\ell, I}- X_{\ell, I}\phiproxi_{\ell})}{\sqrt{N(2\ell+1)}} \right\|_{\infty} \le  \FN\sqrt{\log(p\LN)}\right)  \ge  1 - c_1 e^{- c_2  \log (p  \LN)}. \label{eq:dev_cond_cp}
\end{align}
% Moreover, if the interval $I$ contains no true change point, it holds
% \begin{align}
%     P\left( \bigcap_{\ell = 0}^{\LN-1}\left \|\widehat{\gamma}_{\ell, I} - \widehat{\Gamma}_{\ell, I}\phiproxi_{\ell}\right\|_{\infty}\le  \FN\sqrt{\frac{\log(p\LN)}{N}}\right)  \ge  1 - c_1 e^{- c_2  \log (p \LN)}.\label{eq:dev_cond_nocp}
% \end{align}
\end{proposition}

The final step is to establish a compatibility condition that, conditional on the event \(\mathscr{S}_L\), allows us to bound, for each fixed \(\ell\), both the \emph{prediction error}
\(\|X_{\ell, I}(\widehat{\bm{\phi}}_{\ell, I} - \phiproxi_{\ell})\|_2^2\) 
and the \emph{estimation error}
\(\|\widehat{\bm{\phi}}_{\ell, I} - \phiproxi_{\ell}\|_2^2\) 
by the same deterministic term.

A symmetric matrix \(A \in \mathbb{R}^{p \times p}\) is said to satisfy the \emph{compatibility condition}, also known as the \emph{restricted eigenvalue} (RE) condition, with curvature \(\alpha > 0\) and tolerance \(\tau > 0\) (denoted \(A \sim RE(\alpha, \tau)\)), if for any \(\vartheta \in \mathbb{R}^p\),
\[
\vartheta' A \vartheta \ge \alpha \|\vartheta\|_2^2 - \tau \|\vartheta\|_1^2,
\]
see also \cite{Bickel, Geer2009OnTC}. 

The next proposition provides sufficient conditions under which
\[
\bigcap_{\ell=0}^{\LN-1}\left\{\widehat{\Gamma}_{\ell, I} \sim \operatorname{RE} \left( \alpha_{\ell}, \tau_{\ell} \right)\right\}
\]
holds for some \(\alpha_\ell\)'s and \(\tau_\ell\)'s with probability at least \(1 - c_1 e^{-c_2 \log(pL)}\), for some absolute constants \(c_1, c_2 > 0\).

\begin{proposition}[Compatibility Condition] \label{prop::re_compatibility}
Consider the estimation problem in \eqref{eq:phi_est} and assume that Conditions \ref{cond:identifiability} and \ref{cond:causality} hold.
%Define $\qN = \underset{\ell<\LN }{\max} \qproxi_{\ell}$. 
There exist some constants $c_1,c_2>0$ such that, if
\begin{equation*}\label{eq:suffcond_re}
    N \ge 32 \, \underset{\ell < \LN}{\max}  \left\{ \frac{\omega_\ell^2 }{2\ell+1} \max\left\{ \qproxi_\ell , 1\right\} \right \} \log(p\LN) \qquad \text{with} \qquad \omega_\ell = 54 p  \frac{\mu_{\max; \ell}}{\mu_{\min; \ell}} \frac{\maxCZ}{\minCZ},
\end{equation*}
then 
\begin{equation*}
    P\left(\bigcap_{\ell=0}^{\LN-1}\left\{\widehat{\Gamma}_{\ell, I} \sim \operatorname{RE} \left( \alpha_{\ell}, \tau_{\ell} \right)\right\}\right) \ge 1 - c_1 e^{-c_2 \log(pL)}
\end{equation*}
with
\begin{equation} \label{eq:alpha_tau_def}
    \alpha_{\ell} = \frac{1}{2} \frac{\minCZ}{\mu_{\max; \ell}}, \qquad \text{and} \qquad \tau_{\ell} = \alpha_{\ell}  \omega_\ell^2 \frac{\log(p\LN)}{N(2\ell+1)}.
\end{equation}
\end{proposition}
\begin{remark}
The compatibility condition on $\widehat{\Gamma}_{\ell, I}$ is a requirement on its smallest eigenvalue, which can be seen as a measure of the dependence of the random matrix columns. It ensures that, with high probability, the (sample) minimum eigenvalue of the matrix $\widehat{\Gamma}_{\ell, I}$ is bounded away from zero.
\end{remark}

We are now able to state the main result of this section.
\begin{proposition}[Oracle Inequalities] \label{th::oracle_in}
Consider the estimation problem \eqref{eq:phi_est} and assume Conditions \ref{cond:identifiability} and \ref{cond:causality} hold. Moreover suppose that
$$
\bigcap_{\ell = 0}^{\LN - 1}\widehat{\Gamma}_{\ell, I} \sim \operatorname{RE} \left( \alpha_{\ell}, \tau_{\ell} \right) \quad \text{a.s.}, \quad \text{with} \quad \qproxi_{\ell}\tau_{\ell}\le \alpha_{\ell} /32
$$
and the deviation condition is satisfied almost surely, that is,
$$
\bigcap_{\ell = 0}^{\LN-1}\left \{\left \| \frac{X'_{\ell, I}(\mathbf{Y}_{\ell, I}- X_{\ell, I}\phiproxi_{\ell})}{\sqrt{N(2\ell+1)}} \right\|_{\infty} \le  \FN\sqrt{\log(p\LN)} \right \} \quad \text{a.s..}
$$
Then, if $\lambdaN\ge  4 \FN\sqrt{\log(p\LN)}$, it holds
\begin{align*}
    &\sqrt{N(2\ell+1)} \|\widehat{\bm{\phi}}_{\ell, I} - \phiproxi_{\ell}\|_2 \le 12 \sqrt{\qproxi_{\ell}} \frac{\lambdaN}{\alpha_{\ell} }, \\
    &\sqrt{N(2\ell+1)} \|\widehat{\bm{\phi}}_{\ell, I} - \phiproxi_{\ell}\|_1 \le  48\qproxi_{\ell} \frac{\lambdaN}{\alpha_{\ell}},
\end{align*}
for $\ell = 0,\dots, \LN-1.$
\end{proposition}

\subsection{When $I$ Contains Change Points}
Although not strictly required for the derivation of the main results, the oracle inequalities in Section \ref{sec:oracle_deriv} could be generalized to an interval containing a generic number of change points. We include this argument for completeness.

The first step is to generalize the definition of the $\phiproxi_{\ell}$ to the $p-$dimensional vector satisfying 
$$\sum_{k \in \mathcal{K}} \sum_{t \in \mathcal{H}_k} \mathbb{E}[\bolda_{\ell,0}(t-1) \bolda'_{\ell,0}(t-1) ]\phiproxi_{\ell} = \sum_{k \in \mathcal{K}} \sum_{t \in \mathcal{H}_k}  \mathbb{E}[\bolda_{\ell,0}(t-1) \tilde{\bm{a}}_{\ell,0}^{(k)}{}'(t-1) ]\bm{\phi}_{\ell}^{(k)},$$
with
$$\mathcal{K} = \{k : \bm{\phi}_\ell(t) = \bm{\phi}^{(k)}_{\ell}, \ t = s+p,\dots,e\}, \qquad 
\mathcal{H}_k=[s+p, e]\cap[\eta_k, \eta_{k+1}), \ k \in \mathcal{K}.$$
In matrix notation,
\begin{equation}\label{eq:phiproxi}
  \mathbb{E}[X_{\ell, I}'X_{\ell, I}\phiproxi_{\ell}] = \mathbb{E}[X_{\ell, I}'\mathbf{B}_{\ell, I}],  
\end{equation}
where $\mathbf{B}_{\ell, I}$ is the $N(2\ell+1)-$dimensional vector containing the terms $a_{\ell,m}^{(k)}{}'(t-1)\bm{\phi}_\ell^{(k)}$.
The solution exists and it is unique since
$$
\sum_{k \in \mathcal{K}} \sum_{t \in \mathcal{H}_k} \mathbb{E}[\bolda_{\ell,0}(t-1) \bolda'_{\ell,0}(t-1) ] = N \Gammastar
$$
is invertible.

Note that it holds
\begin{equation}\label{eq:emp_proc_cp}
    \frac{X_{\ell, I}'(\mathbf{Y}_{\ell, I}- X_{\ell, I}\phiproxi_{\ell})}{N(2\ell+1)} = \frac{X_{\ell, I}'\mathbf{E}_{\ell, I}}{N(2\ell+1)} + \frac{X_{\ell, I}'\mathbf{B}_{\ell, I}}{N(2\ell+1)} - \frac{X_{\ell, I}'X_{\ell, I}\phiproxi_{\ell}}{N(2\ell+1)}.
\end{equation}
If no true change point occurs in $I$, $\phiproxi_{\ell}$ coincides with its original definition, since $\bm{\phi}_{\ell}(t)$ is constant over $t$, and the two last terms of \eqref{eq:emp_proc_cp} sum to zero.

We let $\qproxi_{\ell}$ be the number of non-zero entries of $\phiproxi_{\ell}$. However, the components of $\phiproxi_{\ell}$ are weighted averages of the true spectral parameters within the interval $I$, and hence it can occur that $\qproxi_{\ell}=p$ (non-sparse vector) even if the corresponding true sparsity indices are strictly less than $p$ (sparse vectors). This is because we are essentially \emph{mixing} together observations from adjacent segments, making this case much more complex. 

Using this new definition for $\phiproxi_{\ell}$ (and the corresponding $\qproxi_{\ell}$), all the results in \ref{sec:oracle_deriv} hold once we assume
$$
\FN = c_0 a_0 \maxCZ\left(1+ \frac{3+(1+M)pC_\Phi}{\mu_{\min; \ell}} \right),
$$
for some $M>0$.
The additional summands in $\FN$ come from the following inequality
\begin{align*}
\left \| \frac{X'_{\ell, I}(\mathbf{Y}_{\ell, I} - X_{\ell, I}\phiproxi_{\ell})}{N(2\ell+1)} \right\|_{\infty} &\le \left \| \frac{X'_{\ell, I}\mathbf{E}_{\ell, I}}{N(2\ell+1)}  \right\|_\infty \\ &+ \left \|\frac{X'_{\ell, I}\mathbf{B}_{\ell, I}}{N(2\ell+1)}- \mathbb{E}\left[ \frac{X'_{\ell, I}\mathbf{B}_{\ell, I}}{N(2\ell+1)} \right] \right \|_\infty \\ &+\left \| \frac{X'_{\ell, I}X_{\ell, I}\phiproxi_{\ell}}{N(2\ell+1)} -\mathbb{E} \left [ \frac{X'_{\ell, I}X_{\ell, I}\phiproxi_{\ell}}{N(2\ell+1)}  \right]\right\|_\infty \\& = \left\|(I)\right\|_{\infty} + \left\|(II)\right\|_{\infty}+\left\|(III)\right\|_{\infty},
\end{align*}
justified by \eqref{eq:phiproxi}.
Then, similarly to \eqref{eq:devbound4}, we are able to derive deviation bounds for the additional terms $(II)$ and $(III)$.
%\begin{align*}
%\lambda_{\max}(\mathbb{E}[\mathbf{B}_{\ell, I}\mathbf{B}_{\ell, I}']) &\le p \max_{k=0,\dots,K} \|\bm{\phi}^{(k)}_\ell\|_2^2 \max_{k=0,\dots,K}  \lambda_{\max}(\Gamma_\ell^{(k)}) \\ &\le  pC_{\Phi} \max_{k=0,\dots,K}  \lambda_{\max}(\Gamma_\ell^{(k)}),
%\end{align*}
%and
%\begin{align*}
%\lambda_{\max}(\mathbb{E}[(X_{\ell, I}\phiproxi_{\ell})(X_{\ell, I}\phiproxi_\ell)']) &\le p   \|\phiproxi_\ell \|_2^2 \max_{k=0,\dots,K}  \lambda_{\max}(\Gamma_\ell^{(k)}) \\ &\le  pMC_{\Phi} \max_{k=0,\dots,K}  \lambda_{\max}(\Gamma_\ell^{(k)}).
%\end{align*}

In particular, $M$ is a positive constant that satisfies $\|\phiproxi_{\ell}\|^2_2 \le M C_\Phi $. We can show that $M$ exists. Indeed, due to Condition \ref{cond:boundedness} and Lemma \ref{lemma:gamma},
\begin{align*}
    N\|\phiproxi_{\ell}\|_2 &\le \|(\Gammastar)^{-1}\|_{op} \sum_{k \in \mathcal{K}} \sum_{t \in \mathcal{H}_k} \| \mathbb{E}[\bolda_{\ell,0}(t-1) \bm{a}_{\ell,0}^{(k)}{}'(t-1) ]\|_{op} \|\bm{\phi}_{\ell}^{(k)}\|_2\\
    &\le \frac{
\sqrt{C_\Phi}}{2\pi\mathfrak{m}_\ell} \sum_{k \in \mathcal{K}} \sum_{t \in \mathcal{H}_k} \| \mathbb{E}[\bolda_{\ell,0}(t-1) \bm{a}_{\ell,0}^{(k)}{}'(t-1) ]\|_{op}.
\end{align*}
Moreover, for fixed $k$ and $t$, we can consider the following cases:
\begin{enumerate}
    \item all the entries of $\bolda_{\ell,0}(t-1)$ belong to a segment different from the $k-$th one;
    \item all the entries of $\bolda_{\ell,0}(t-1)$ belong to the $k-$th segment;
    \item $\bolda_{\ell,0}(t-1)$ has entries that belong to both the $(k-1)-$th and $k-$th segments.
\end{enumerate}
It is readily seen that $\mathbb{E}[\bolda_{\ell,0}(t-1) \bm{a}_{\ell,0}^{(k)}{}'(t-1) ]$ is the zero matrix in the first case, and $\Gamma_\ell^{(k)}$ in the second case.
In the last case, depending on how many entries of $\bolda_{\ell,0}(t-1)$ belong to the $k-$th segments, say $j$, the first $p-j$ rows of $\mathbb{E}[\bolda_{\ell,0}(t-1) \bm{a}_{\ell,0}^{(k)}{}'(t-1) ]$ contain only zeros, while the last $j$ rows coincide with the last $j$ rows of $\Gamma_\ell^{(k)}$.
Hence, computing $\| \mathbb{E}[\bolda_{\ell,0}(t-1) \bm{a}_{\ell,0}^{(k)}{}'(t-1) ]\|_{op}$ is equivalent to computing the operator norm of the reduced $\Gamma_{\ell}^{(k)}$ matrix (of dimension $j\times p$). Then, by Theorem 1 in \cite{THOMPSON19721}, we have that
\begin{align*}
\| \mathbb{E}[\bolda_{\ell,0}(t-1) \bm{a}_{\ell,0}^{(k)}{}'(t-1) ]\|_{op}
&\le  \| \mathbb{E}[\bm{a}_{\ell,0}^{(k)}(t-1) \bm{a}_{\ell,0}^{(k)}{}'(t-1) ]\|_{op}\le 2\pi \mathcal{M}_\ell.
\end{align*}
Moreover, since
$$
0 < \frac{\mathcal{M}_\ell}{\mathfrak{m}_\ell} \le \frac{\maxCZ }{\minCZ } \frac{\mu_{\max ;\ell}}{\mu_{\min ;\ell}} \le \sqrt{M},
$$
for some positive constant $M$, we obtain the result.
See also Remark 4.1 in \cite{CAPONERA2021167}.\\

%$$\phiproxi_{\ell} = \sum_{k: I_k \ne \emptyset} \frac{|I_k|}{N} \bm{\phi}^{(k)}_{\ell}.$$

%$$\phiproxi_{\ell} = \Gammastar^{-1} \sum_{t = s+p}^e\Gamma_{\ell, p}(t)\bm{\phi}_{\ell; t}$$

\section{Proofs}\label{sec:proofs}
In this subsection we collect all the proofs of the previous results. Note that we make use of the matrix notation defined in Appendix \ref{sec:est_singleinterval}.

\subsection{Proofs of Appendix \ref{sec:est_singleinterval}}

\subsubsection{Lemma \ref{lemma:gamma}}
First observe that
$$
\widehat{\Gamma}_{\ell, I} = \frac{X'_{\ell, I}X_{\ell, I}}{N(2\ell+1)} = \frac{1}{N(2\ell+1)} \sum_{t=s+p}^{e} \sum_{m=-\ell}^\ell \bolda_{\ell,m}(t-1) \bolda'_{\ell,m}(t-1),
$$
and hence
$$
\Gammastar = \mathbb{E}\left[\widehat{\Gamma}_{\ell, I} \right] = \frac{1}{N} \sum_{t=s+p}^{e} \mathbb{E}[\bolda_{\ell,0}(t-1) \bolda'_{\ell,0}(t-1)].
$$
Depending on the change points in $I$, each matrix $\mathbb{E}[\bolda_{\ell,0}(t-1) \bolda'_{\ell,0}(t-1)]$, $t=s+p,\dots,e$, has a block diagonal structure, and the blocks related to change point $\eta_k$ are principal submatrices of $\Gamma_{\ell}^{(k)}$.
Then, by the interlacing theorem \cite{THOMPSON19721},
$$
\min_{k=0,\dots,K} \lambda_{\min} (\Gamma_\ell^{(k)}) \le \lambda_{\min} (\Gammastar) \le \lambda_{\max} (\Gammastar) \le \max_{k=0,\dots,K} \lambda_{\max} (\Gamma_\ell^{(k)}),
$$
and the result follows.

\subsubsection{Proposition \ref{prop::dev_bound}: Deviation Bounds}
Consider $J = \text{supp}(v) = \{j_1, \dots, j_r\} \subset \{1, \dots , p\},$ $r \ge 1.$ In addition, define
$$
V_J = X_{\ell, I}v = \sum_{j \in J}v_j\mathbf{Y}_{\ell, I}(j)\sim N(0_{N(2\ell+1)\times 1 }, Q_J)
$$
where
$$
Q_J = \mathbb{E}[V_JV_J']= \sum_{j,j'=1}^p v_j v_{j'} \mathbb{E}[\mathbf{Y}_{\ell, I}(j)\mathbf{Y}'_{\ell, I}(j')].
$$
Note that $Q_J$ is a $N(2\ell+1) \times N(2\ell+1)$ block diagonal matrix with $2\ell+1$ equal blocks, due to independence and equality in distribution of the harmonic coefficients for a given multipole $\ell$.

Now, observe that
\begin{align*}
    & v'\widehat{\Gamma}_{\ell, I}v=\frac{ v' X_{\ell, I}'X_{\ell, I}v}{N(2\ell +1)}= \frac{ V_J'V_J}{N(2\ell +1)}= \frac{ Z'Q_JZ}{N(2\ell +1)}, \quad Z \sim N(0, I_{N(2\ell +1)}),\\
    & v'\Gammastar v= E\left[\frac{ Z'Q_JZ}{N(2\ell +1)}\right].
\end{align*}
By the Hanson-Wright inequality, we get 
\begin{align}
    P\left(\left|v'(\widehat{\Gamma}_{\ell, I}-\Gammastar)v\right|> \zeta\right) &=
    P\left(\left|Z'Q_JZ - E[Z'Q_JZ]\right|> \zeta N(2\ell +1) \right) \nonumber \\
    & \le 2 \exp\left[-c \,\min\left\{\frac{N^2(2\ell +1)^2\zeta^2}{\|Q_J\|^2_F}, \frac{N(2\ell +1) \zeta}{\|Q_J\|_{op}}\right\}\right]
\end{align}
Since $\|Q_J\|^2_F/N(2\ell +1) \le \|Q_J\|_{op}^2$, setting $\zeta = \eta \|Q_J\|_{op}$, we have
\begin{align*}
    P\left(\left|v'(\widehat{\Gamma}_{\ell, I}-\Gammastar)v\right|> \eta \|Q_J\|_{op}\right) \le 2 \exp\left[-c N(2\ell +1)\,\min\{\eta^2, \eta\}\right]
\end{align*}

For $\omega \in \mathbb{R}^{N(2\ell+1)}$, $\|\omega\| = 1$,
\begin{align*}
w' Q_J w  & =  (w \otimes v)' \mathbb{E}\begin{bmatrix} \mathbf{Y}_{\ell, I} (1) \mathbf{Y}'_{\ell, I} (1) & \mathbf{Y}_{\ell, I} (1) \mathbf{Y}'_{\ell, I} (2) &\cdots& \mathbf{Y}_{\ell, I} (1) \mathbf{Y}'_{\ell, I} (p)\\
\mathbf{Y}_{\ell, I} (2) \mathbf{Y}'_{\ell, I} (1) & \mathbf{Y}_{\ell, I} (2) \mathbf{Y}'_{\ell, I} (2) & \cdots& \mathbf{Y}_{\ell, I} (2) \mathbf{Y}'_{\ell, I} (p) \\
\vdots & \vdots& \ddots & \vdots\\
\mathbf{Y}_{\ell, I} (p) \mathbf{Y}'_{\ell, I} (1) & \mathbf{Y}_{\ell, I} (p) \mathbf{Y}'_{\ell, I} (2) & \cdots & \mathbf{Y}_{\ell, I} (p) \mathbf{Y}'_{\ell, I} (p)
\end{bmatrix}
(w \otimes v)\\
& \le \sum_{j\in J} \Lambda_{\max} (\mathbb{E}[\mathbf{Y}_{\ell, I} (j) \mathbf{Y}'_{\ell, I} (j)]).
%= \sum_{r,r'=s+p}^e \sum_{j,j'=1}^p w_r w_{r'} v_j v_{j'} \mathbb{E}[a_{\ell,m}(r-j)a_{\ell,m} (r'-j')]
\end{align*}
Depending on the change points, each matrix $\mathbb{E}[\mathbf{Y}_{\ell, I} (j) \mathbf{Y}'_{\ell, I} (j)]$, $j \in J$, has a block diagonal structure, and the blocks related to change point $\eta_k$ are principal submatrices of the $N\times N$ matrix with generic $ij-$th element $C^{(k)}_\ell (i-j)$ (whose largest eigenvalue is bounded by $2\pi \mathcal{M}(f^{(k)}_\ell)$).
Then, by the interlacing theorem \cite{THOMPSON19721},
$$
\sum_{j \in J} \Lambda_{\max} (\mathbb{E}[\mathbf{Y}_{\ell, I} (j) \mathbf{Y}'_{\ell, I} (j)]) \le 2\pi  |J| \sup_{k=0,\dots,K}  \mathcal{M}(f^{(k)}_\ell),
$$
and we obtain $\|Q_J\|_{op} \le 2\pi r \mathcal{M}_\ell$.

We can then conclude that 
$$
P\left(\left|v'(\widehat{\Gamma}_{\ell, I}-\Gammastar)v\right|> 2 \pi r \mathcal{M}_{\ell} \eta\right) \le 2 \exp\left[-c \,N\, (2\ell+1) \,\min\{\eta^2, \eta\}\right],
$$
proving \eqref{eq:devbound1}.

To prove \eqref{eq:devbound2}, we note that
\begin{align}\label{eq::uvdecomp}
2\left|u'(\widehat{\Gamma}_{\ell, I}-\Gammastar)v\right| &\le \left|u'(\widehat{\Gamma}_{\ell, I}-\Gammastar)u\right| + \left|v'(\widehat{\Gamma}_{\ell, I}-\Gammastar)v\right|\notag \\&+ \left|(u+v)'(\widehat{\Gamma}_{\ell, I}-\Gammastar)(u+v)\right|
\end{align}
and $u+v$ is at most $2r-$sparse with $\|u+v\|\le 2$. The result follows by applying \eqref{eq:devbound1} separately on each of the three terms on the right-hand side.
The element-wise deviation bound \eqref{eq:devbound3} is obtained by choosing $u = e_i, v = e_j$.

Similarly, we prove \eqref{eq:devbound4}. For a fixed $h=1,\dots,p$, we consider the $N(2\ell+1) \times 2$ matrix
$$
\widetilde{X}_{\ell, I} = \{\mathbf{Y}_{\ell, I}(h) : \mathbf{E}_{\ell, I}\},
$$
and we observe that 
$$
\frac{\mathbf{Y}'_{\ell, I}(h) \mathbf{E}_{\ell, I}}{N(2\ell+1)} = \frac{u'\widetilde{X}'_{\ell, I} \widetilde{X}_{\ell, I}v}{N(2\ell+1)}, \qquad \mathbb{E}\left [\frac{\mathbf{Y}'_{\ell, I}(h) \mathbf{E}_{\ell, I}}{N(2\ell+1)} \right] = 0,
$$
with $u=(1,0)'$, $v=(0,1)'$. Moreover,
\begin{align*}
&\lambda_{\max}(\mathbb{E}[\mathbf{Y}_{\ell, I}(h) \mathbf{Y}'_{\ell, I}(h)]) \le 2\pi \mathcal{M}_\ell \le \frac{\maxCZ}{\mu_{\min;\ell}}, \\ &\lambda_{\max}(\mathbb{E}[\mathbf{E}_{\ell, I} \mathbf{E}'_{\ell, I}]) \le 2\pi \mathcal{M}_{\ell;Z} = \maxCZ.
\end{align*}
Subsequently, we derive
\begin{align*}  
    &P\left(\left|\frac{\mathbf{Y}'_{\ell, I}(h) \mathbf{E}_{\ell, I}}{N(2\ell+1)}\right|> a_0 \maxCZ\left(1+ \frac{1}{\mu_{\min; \ell}} \right)
    \eta\right) \le 6 \exp\left[-c \,N\, (2\ell+1) \,\min\{\eta^2, \eta\}\right]
\end{align*}
with $a_0=3/2$, by employing the inequality \eqref{eq::uvdecomp}. The procedure involves iterating the steps required to establish \eqref{eq:devbound1} independently for each of the three terms on the right-hand side.
The final result is then obtained by observing that
$$
\left \| \frac{X'_{\ell, I}\mathbf{E}_{\ell, I}}{N(2\ell+1)} \right\|_{\infty} = \underset{1\le h\le p}{\max} \left | \frac{\mathbf{Y}'_{\ell, I}(h)\mathbf{E}_{\ell, I}}{N(2\ell+1)} \right|.
$$

% Thus it follows that
% \begin{align*}  
%     &P\left(\underset{1\le h\le p}{\max}\left|\frac{\mathbf{Y}'_{\ell, I}(h) \mathbf{E}_{\ell, I}}{N(2\ell+1)}\right|>  a_0 \maxCZ\left(1+ \frac{1}{\mu_{\min; \ell}} \right)
%     \eta\right) \le 6p \exp\left[-c \,N\, (2\ell+1) \,\min\{\eta^2, \eta\}\right]
% \end{align*}

\subsubsection{Proposition \ref{prop::basic_inequality}: Basic Inequality}

Since the vector $\widehat{\bm{\phi}}_{\ell, I}$ is the solution of the minimization problem \eqref{eq:phi_est}, it holds that
\begin{equation*}
   \|\mathbf{Y}_{\ell, I}- X_{\ell, I}\widehat{\bm{\phi}}_{\ell, I} \|^2_2 + \lambda_\ell \sqrt{N(2\ell+1)}\|\widehat{\bm{\phi}}_{\ell, I}\|_1 \le \|\mathbf{Y}_{\ell, I}- X_{\ell, I}\phiproxi_{\ell} \|^2_2 + \lambda_\ell \sqrt{N(2\ell+1)}\|\phiproxi_{\ell}\|_1.
\end{equation*}
Letting $\bm{\Delta}_\ell = \widehat{\bm{\phi}}_{\ell, I} - \phiproxi_\ell$, we observe that
$$
\|\mathbf{Y}_{\ell, I} \pm X_{\ell, I}\phiproxi_\ell - X_{\ell, I}\widehat{\bm{\phi}}_{\ell, I} \|^2_2 = \|\mathbf{Y}_{\ell, I} - X_{\ell, I}\phiproxi_\ell  \|^2_2 + \| X_{\ell, I}\bm{\Delta}_\ell  \|^2_2 - 2 \bm{\Delta}'_\ell X_{\ell, I} (\mathbf{Y}_{\ell, I} - X_{\ell, I}\phiproxi_\ell ) .
$$
Hence, we have
\begin{align*}
\bm{\Delta}'_{\ell}X'_{\ell, I}X_{\ell, I}\bm{\Delta}_\ell  -2\bm{\Delta}'_{\ell}X'_{\ell, I}(\mathbf{Y}_{\ell, I} - X_{\ell, I}\phiproxi_\ell)  \le \lambda_\ell \sqrt{N(2\ell+1)} (\|\phiproxi_{\ell}\|_1 -\|\widehat{\bm{\phi}}_{\ell, I}\|_1), \end{align*}
which leads to the final result
\begin{align*}
\bm{\Delta}'_{\ell}\widehat{\Gamma}_{\ell, I}\bm{\Delta}_\ell \le \frac{2\bm{\Delta}'_{\ell}X'_{\ell, I}(\mathbf{Y}_{\ell, I} - X_{\ell, I}\phiproxi_\ell)}{N(2\ell+1)} + \frac{\lambda_\ell}{\sqrt{N(2\ell+1)}} \left( \|\phiproxi_{\ell}\|_1 - \|\phiproxi_{\ell} +\bm{\Delta}_{\ell}\|_1 \right).
\end{align*}

\subsubsection{Proposition \ref{prop::dev_condition}: Deviation Condition}

First note that, since $I$ does not contain change points,
$$X_{\ell, I}'(\mathbf{Y}_{\ell, I}- X_{\ell, I}\phiproxi_{\ell}) = X_{\ell, I}'\mathbf{E}_{\ell, I}.$$
Then, by \eqref{eq:devbound4},
\begin{align*}
    &P \left (\left \|\frac{X'_{\ell, I}(\mathbf{Y}_{\ell, I} - X_{\ell, I}\phiproxi_{\ell})}{N(2\ell+1)} \right \|_\infty > a_0 \maxCZ\left( \frac{1}{\mu_{\min; \ell}} \right)
    \eta\right) \\
    &\le 6 p  \exp\left[-c \,N\, (2\ell+1) \,\min\{\eta^2, \eta\}\right],
\end{align*}
We then choose $\eta=\eta_\ell = c_0 \sqrt{\frac{\log(p\LN)}{N(2\ell+1)}}$, and since
$$
 = \FN = c_0 a_0 \maxCZ\left(1+ \frac{3+(1+M)pC_\Phi}{\mu_{\min; \ell}} \right) 
,$$
we have 
\begin{align*}  
    &P\left(\left \|\frac{X'_{\ell, I}(\mathbf{Y}_{\ell, I} - X_{\ell, I}\phiproxi_{\ell})}{\sqrt{N(2\ell+1)}} \right \|_\infty> \FN  \sqrt{\log(p\LN)} \right) \le 6 p \exp\left[-c \min\{c_0^2, c_0\} \log(p\LN)\right].
\end{align*}
Hence,
\begin{align*}
    &P\left( \bigcap_{\ell = 0}^{\LN-1}\left \|\frac{X'_{\ell, I}(\mathbf{Y}_{\ell, I} - X_{\ell, I}\phiproxi_{\ell})}{\sqrt{N(2\ell+1)}} \right \|_\infty \le  \FN\sqrt{\log(p\LN)}\right) \\
    =&P\left( \underset{\ell <\LN}{\max}\left \|\frac{X'_{\ell, I}(\mathbf{Y}_{\ell, I} - X_{\ell, I}\phiproxi_{\ell})}{\sqrt{N(2\ell+1)}} \right \|_\infty  \le  \FN\sqrt{\log(p\LN)}\right) \\
    \ge &  1 - 6 p\LN \exp\left[-c  \,\min\{c_0^2, c_0\}\log(p\LN)\right]\\
    =&  1 - 6 \exp\left[- (c  \,\min\{c_0^2, c_0\}-1)\log(p\LN)\right]
\end{align*}
and the statement is proved with $c_1 = 6, \, c_2 = c  \,\min\{c_0^2, c_0\}-1$, where $c_0$ is any positive constant that satisfies $c_2>0$.

\subsubsection{Proposition \ref{prop::re_compatibility}: Compatibility Condition}
We want to prove that
$$
P\left(\bigcap_{\ell=0}^{\LN-1}A_{\ell}\right) \ge 1 - c_1 e^{-c_2 \log(pL)},
$$
where $A_{\ell} = \left\{v_{\ell}'\widehat{\Gamma}_{\ell, I}v_{\ell}\ge \alpha_{\ell} \|v_{\ell}\|_2^2 -\tau _{\ell} \|v_{\ell}\|_1^2, \ \forall v_\ell \in \mathbb{R}^p \right\}$.
Starting from \eqref{eq:devbound1}, it is known that
$$
P\left(\left|v_{\ell}'(\widehat{\Gamma}_{\ell, I}-\Gammastar)v_{\ell}\right|>2\pi r \mathcal{M}_\ell\eta\right) \le 2 \exp\left[-c \,N\, (2\ell+1) \,\min\{\eta^2, \eta\}\right]
$$
and taking into account that
$$
2\pi r \mathcal{M}_\ell \le p \ \frac{\maxCZ}{\mu_{\min; \ell}}, \qquad r\ge1,
$$
we get the following result
$$
P\left(\left|v_{\ell}'(\widehat{\Gamma}_{\ell, I}-\Gammastar)v_{\ell}\right|> p \ \frac{\maxCZ}{\mu_{\min; \ell}} \eta\right) \le 2 \exp\left[-c \,N\, (2\ell+1) \,\min\{\eta^2, \eta\}\right].
$$
Using Lemma F.2 of the supplementary material of \cite{basu} we obtain
\begin{align*}
    &P\left(\underset{v_{\ell} \in \mathcal{K}(2s_\ell)}{\sup} \left|v_{\ell}'(\widehat{\Gamma}_{\ell, I}-\Gammastar)v_{\ell}\right|> p \ \frac{\maxCZ}{\mu_{\min; \ell}} \eta\right) \\ & \le 2 \exp\left[-c \,N\, (2\ell+1) \,\min\{\eta^2, \eta\}+ 2s_\ell \min \left \{\log p, \log \left(  \frac{21 e p}{2s_\ell}\right)\right\}\right]
\end{align*}
where $ \mathcal{K}(2s)= \left \{v \in \mathbb{R}^p : \|v\|_2 \le 1, \|v\|_0 \le 2s\right\}$, for an integer $s\ge 1$.
Setting 
$$
 \eta = \left(\omega_{\ell}\right)^{-1} = \frac{1}{54 \ p}\frac{\mu_{\min; \ell}}{\mu_{\max; \ell}} \frac{\minCZ}{\maxCZ}
$$
we obtain
$$
P\left(\underset{v_{\ell} \in \mathcal{K}(2s_\ell)}{\sup} \left|v_{\ell}'(\widehat{\Gamma}_{\ell, I}-\Gammastar)v_{\ell}\right|\le\frac{1}{54}\frac{\minCZ}{\mu_{\max; \ell}}\right) \ge1- 2 e^{-c \,N\, (2\ell+1) \,
\omega_{\ell}^{-2} + 2s_\ell \min \left \{\log p, \log \left(  \frac{21 e p}{2s_\ell}\right)\right\}}
$$
and applying Lemma 12 of the supplementary material of \cite{LW12} it results in
\begin{align*}
    &P\left( \left|v_{\ell}'(\widehat{\Gamma}_{\ell, I}-\Gammastar)v_{\ell}\right|\le\frac{1}{2}\frac{\minCZ}{\mu_{\max; \ell}}\left\{\|v_{\ell}\|_2^2+ \frac 1 s \|v_{\ell}\|_1^2\right\}, \ \forall v_{\ell} \in \mathbb{R}^p\right) \\ &\qquad\ge1- 2 e^{-c \,N\, (2\ell+1) \, \omega_\ell^{-2} + 2s_\ell \min \left \{\log p, \log \left(  \frac{21 e p}{2s_\ell}\right)\right\}}.
\end{align*}
Moreover, it holds
\begin{align*}
    \left|v_{\ell}'(\widehat{\Gamma}_{\ell, I} -\Gammastar) v_{\ell}\right| & = \left|v_{\ell}' \widehat{\Gamma}_{\ell, I}v_{\ell} - v_{\ell}' \Gammastar v_{\ell}\right| = \left|v_{\ell}' \Gammastar v_{\ell} - v_{\ell}' \widehat{\Gamma}_{\ell, I}v_{\ell} \right| \ge \left|v_{\ell}' \Gammastar v_{\ell} \right |- \left | v_{\ell}' \widehat{\Gamma}_{\ell, I}v_{\ell} \right|\\
    & \ge \Lambda_{\min}\left(\Gammastar\right) \| v_{\ell}\|^2_2 - \left | v_{\ell}' \widehat{\Gamma}_{\ell, I}v_{\ell} \right| \ge \frac{\minCZ}{\mu_{\max; \ell}} \| v_{\ell}\|^2_2 - \left | v_{\ell}' \widehat{\Gamma}_{\ell, I}v_{\ell} \right| \\
    &=\frac{\minCZ}{\mu_{\max; \ell}} \| v_{\ell}\|^2_2 -  v_{\ell}' \widehat{\Gamma}_{\ell, I}v_{\ell}, 
\end{align*}
which implies that
\begin{align*}
    &P\left( \frac{\minCZ}{\mu_{\max; \ell}} \| v_{\ell}\|^2_2 -  v_{\ell}' \widehat{\Gamma}_{\ell, I}v_{\ell}\le\frac{1}{2}\frac{\minCZ}{\mu_{\max; \ell}}\left\{\|v_{\ell}\|_2^2+ \frac{1}{s_\ell} \|v_{\ell}\|_1^2\right\}, \ \forall v_{\ell} \in \mathbb{R}^p\right) \\ &\qquad\ge1- 2 e^{-c \,N\, (2\ell+1) \,\omega_\ell^{-2}+ 2s_\ell \min \left \{\log p, \log \left(  \frac{21 e p}{2s_\ell}\right)\right\}}.
\end{align*}
Hence, we can rearrange the terms in the previous relation to have
\begin{align*}
    &P\left(  v_{\ell}' \widehat{\Gamma}_{\ell, I}v_{\ell}\ge \frac{1}{2} \frac{\minCZ}{\mu_{\max; \ell}} \| v_{\ell}\|^2_2 - \frac{1}{2s_\ell} \frac{\minCZ}{\mu_{\max; \ell}}    \|v_{\ell}\|_1^2, \ \forall v_{\ell} \in \mathbb{R}^p\right) \\ &\qquad\ge1- 2 e^{-c \,N\, (2\ell+1) \,\omega_\ell^{-2}+ 2s_\ell \min \left \{\log p, \log \left(  \frac{21 e p}{2s_\ell}\right)\right\}},
\end{align*}
and we set
$$
s_\ell = \frac{c N (2\ell+1) \omega_\ell^{-2}}{4 \log (p \LN)} \ge 1, \qquad c = \frac{1}{2e},
$$
so that
\begin{align*}
    &P\left(  v_{\ell}' \widehat{\Gamma}_{\ell, I}v_{\ell}\ge \frac{1}{2} \frac{\minCZ}{\mu_{\max; \ell}} \| v_{\ell}\|^2_2 - \frac{1}{2s_\ell} \frac{\minCZ}{\mu_{\max; \ell}}    \|v_{\ell}\|_1^2, \ \forall v_{\ell} \in \mathbb{R}^p\right) \\ &\qquad\ge1- 2 e^{-c \,N\, (2\ell+1) \,\omega_\ell^{-2}\left ( 1- \frac{1}{2} \frac{\log p}{\log(pL)} \right)},
\end{align*}
Now, since by assumption $c N(2\ell+1) \omega_\ell^{-2} \ge 4 \log (pL)$ for any $\ell < L$, we obtain
\begin{align*}
    P\left(\bigcup_{\ell=0}^{\LN-1}\overline{A}_{\ell}\right) & \le \sum_{\ell=0}^{\LN-1} \left[1 - P \left(A_{\ell}\right)\right]\\
    &\le 2 L e^{- 4 \log(pL) \left ( 1- \frac{1}{2} \frac{\log p}{\log(pL)} \right)}\\
    & = 2 e^{-4\log(pL) + 2\log p  + \log L}\\
    & \le 2  e^{-2 \log(pL)}
\end{align*}
which concludes the proof.

\subsubsection{Proposition \ref{th::oracle_in}: Oracle Inequalities}
The goal of this lemma is to obtain an upper bound for $\|\Delta_\ell\|_2$ and $\|\Delta_\ell\|_1$.

The basic inequality in \eqref{eq:basic_ineq} states that 
\begin{align*}
\sqrt{N(2\ell+1)}\bm{\Delta}'_{\ell}\widehat{\Gamma}_{\ell, I}\bm{\Delta}_\ell &\le \frac{2\bm{\Delta}'_{\ell}X'_{\ell, I}(\mathbf{Y}_{\ell, I} - X_{\ell, I}\phiproxi_\ell)}{\sqrt{N(2\ell+1})} + \lambdaN \left( \|\phiproxi_{\ell}\|_1 - \|\phiproxi_{\ell} +v_{\ell}\|_1 \right),
\end{align*}
and hence, applying the deviation condition in \eqref{eq:dev_cond_cp}, it holds that 
\begin{align}
\sqrt{N(2\ell+1)}\bm{\Delta}'_{\ell}\widehat{\Gamma}_{\ell, I}\bm{\Delta}_\ell &\le 2 \|\bm{\Delta}_{\ell}\|_1 \left \| \frac{X'_{\ell, I}(\mathbf{Y}_{\ell, I} - X_{\ell, I}\phiproxi_\ell)}{\sqrt{N(2\ell+1)}} \right\|_\infty+ \lambdaN \left( \|\phiproxi_{\ell}\|_1 - \|\phiproxi_{\ell} +v_{\ell}\|_1 \right)\nonumber\\
&\le 2 \|\bm{\Delta}_{\ell}\|_1 \FN \sqrt{\log(p\LN)} + \lambdaN \left( \|\phiproxi_{\ell}\|_1 - \|\phiproxi_{\ell} +v_{\ell}\|_1 \right).\label{eq:proof_res}
\end{align} 
Now, let $J = \text{supp}\left(\phiproxi_{\ell}\right)=\left \{j_{1}, \ \dots \ , j_{\qproxi_{\ell}}\right \}$ be such that $|J| = \qproxi_{\ell}$, then $J^c = \{1,  \dots,  p\} \setminus J$, $\|\phiproxi_{\ell, J}\|_1 = \|\phiproxi_{\ell}\|_1$ and $\|\phiproxi_{\ell, J^c}\|_1=0.$ Consequently, it holds that
\begin{align*}
    \|\phiproxi_{\ell}+ \bm{\Delta}_{\ell}\|_1 &= \|\phiproxi_{\ell, J}+ \bm{\Delta}_{\ell,J}\|_1 +\|\bm{\Delta}_{\ell,J^c}\|_1 \\
    & \ge  \|\phiproxi_{\ell, J}\|_1 - \|\bm{\Delta}_{\ell,J}\|_1 +\|\bm{\Delta}_{\ell,J^c}\|_1 
\end{align*}
which implies
\begin{align*}
    \lambdaN \left( \|\phiproxi_{\ell}\|_1 - \|\phiproxi_{\ell}+ \bm{\Delta}_{\ell}\|_1 \right) & \le \lambdaN \left( \|\phiproxi_{\ell}\|_1 - \|\phiproxi_{\ell, J}\|_1 + \|\bm{\Delta}_{\ell,J}\|_1 -\|\bm{\Delta}_{\ell,J^c}\|_1  \right)\\ & \le \lambdaN \left( \|\bm{\Delta}_{\ell,J}\|_1 -\|\bm{\Delta}_{\ell,J^c}\|_1  \right).
\end{align*}
Having explicitly required that $ \lambdaN \ge 4 \FN\sqrt{\log(p\LN)}$, we obtain
\begin{align}
    0 \le \sqrt{N(2\ell+1)}\bm{\Delta}'_{\ell}\widehat{\Gamma}_{\ell, I}\bm{\Delta}_{\ell}& \le \frac{\lambdaN }{2}\|\bm{\Delta}_{\ell}\|_1 + \lambdaN \left( \|\bm{\Delta}_{\ell,J}\|_1 -\|\bm{\Delta}_{\ell,J^c}\|_1  \right) \nonumber\\
    & = \frac{\lambdaN }{2}\left( \|\bm{\Delta}_{\ell,J}\|_1 +\|\bm{\Delta}_{\ell,J^c}\|_1  \right) + \lambdaN \left( \|\bm{\Delta}_{\ell,J}\|_1 -\|\bm{\Delta}_{\ell,J^c}\|_1  \right) \nonumber\\ & = \frac{3\lambdaN }{2}\|\bm{\Delta}_{\ell,J}\|_1- \frac{\lambdaN }{2}\|\bm{\Delta}_{\ell,J^c}\|_1 \le \frac{3\lambdaN }{2}\|\bm{\Delta}_{\ell}\|_1. \label{eq:proof4.12_1}
\end{align}
This ensures that $\|\bm{\Delta}_{\ell,J^c}\|_1 \le 3\|\bm{\Delta}_{\ell,J}\|_1 $ and hence, adding $\|\bm{\Delta}_{\ell,J}\|_1 $ on both sides, that $\|\bm{\Delta}_{\ell}\|_1 \le 4\|\bm{\Delta}_{\ell,J}\|_1 $, which implies that 
$$
\|\bm{\Delta}_{\ell}\|_1 \le 4 \sqrt{\qproxi_{\ell}} \|\bm{\Delta}_{\ell}\|_2,
$$
from Cauchy-Schwartz inequality.\\
Now we use this property into the (RE) inequality in Proposition \ref{prop::re_compatibility}, keeping in mind that we specifically required that $\qproxi_{\ell}\tau_{\ell}\le \alpha_{\ell} /32$, and we obtain

\begin{align}
    \bm{\Delta}'_{\ell}\widehat{\Gamma}_{\ell, I}\bm{\Delta}_{\ell}& \ge \alpha_{\ell} \|\bm{\Delta}_{\ell}\|_2^2 -\tau _{\ell} \|\bm{\Delta}_{\ell}\|_1^2\ge \alpha_{\ell} \|\bm{\Delta}_{\ell}\|_2^2 - 16\qproxi_{\ell}\tau _{\ell} \|\bm{\Delta}_{\ell}\|_2^2 \nonumber\\
    & \ge \alpha_{\ell} \|\bm{\Delta}_{\ell}\|_2^2 - \frac{\alpha_{\ell}}{2} \|\bm{\Delta}_{\ell}\|_2^2 \ge \frac{\alpha_{\ell}}{2} \|\bm{\Delta}_{\ell}\|_2^2 .\label{eq:proof4.12_2}
\end{align}

Hence, combining \eqref{eq:proof4.12_1} and \eqref{eq:proof4.12_2}, we get
$$
\sqrt{N(2\ell+1}) \frac{\alpha_{\ell}}{2} \|\bm{\Delta}_{\ell}\|_2^2 \le \sqrt{N(2\ell+1}) \bm{\Delta}'_{\ell}\widehat{\Gamma}_{\ell, I}\bm{\Delta}_{\ell} \le \frac{3 \lambdaN}{2}\|\bm{\Delta}_{\ell}\|_1 \le 6\sqrt{\qproxi_{\ell}}\lambdaN\|\bm{\Delta}_{\ell}\|_2,
$$
which results in the following estimate for the norm of the error
$$
\sqrt{N(2\ell+1)}\frac{\alpha_{\ell}}{3} \|\bm{\Delta}_{\ell}\|_2^2 \le \lambdaN\|\bm{\Delta}_{\ell}\|_1 \le \lambdaN 4 \sqrt{\qproxi_{\ell}}\|\bm{\Delta}_{\ell}\|_2.
$$
As consequence,
\begin{align}
    &\sqrt{N(2\ell+1)} \|\bm{\Delta}_{\ell}\|_2 \le 12 \sqrt{\qproxi_\ell} \frac{\lambdaN}{\alpha_\ell} , \\
    &\sqrt{N(2\ell+1)}\|\bm{\Delta}_{\ell}\|_1 \le 4 \sqrt{\qproxi_{\ell} } \sqrt{N(2\ell+1)} \|\bm{\Delta}_{\ell}\|_2 \le  48\qproxi_\ell \frac{\lambdaN}{\alpha_\ell} , \label{eq:bound_l1} \\
    & N(2\ell+1)\bm{\Delta}'_{\ell}\widehat{\Gamma}_{\ell, I}\bm{\Delta}_{\ell} \le\frac{3 \lambdaN}{2} \sqrt{N(2\ell+1)}\|\bm{\Delta}_{\ell}\|_1 \le 72 \qproxi_\ell \frac{\lambdaN^2}{\alpha_\ell}.
\end{align}

\subsection{Proof of Theorem \ref{th::error_diff}: Prediction Error}

We are considering an interval $I$ that contains no true change point. Let $\bm{\Delta}_\ell = \bm{\widehat{\phi}}_{\ell, I}-\phiproxi_\ell$, for every fixed $\ell = 0, \dots, \LN-1$. 

Moreover, let 
$$\mathcal{N}_L = \{\ell < L : N(2\ell+1) \ge 32 \omega_\ell^2 \max \{\qproxi_\ell, 1\} \log(pL)\},$$ and 
$$\overline{\mathcal{N}}_L = \{\ell < L : N(2\ell+1) < 32 \omega_\ell^2 \max \{\qproxi_\ell, 1\} \log(pL)\}$$

For every fixed $\ell\in \mathcal{N}_L$, starting from the definition of $\bm{\widehat{\phi}}_{\ell, I}$ and using the result obtained in Proposition \ref{th::oracle_in}, we get the following chain of inequalities
\begin{align*}
   \|\mathbf{Y}_{\ell, I} - X_{\ell, I}\bm{\widehat{\phi}}_{\ell, I}\|_2^2 -  \left \| \mathbf{Y}_{\ell, I} - X_{\ell, I}\phiproxi_{\ell} \right\|_2^2  &\le - \lambdaN \sqrt{N(2\ell+1)}\|\bm{\widehat{\phi}}_{\ell, I}\|_1 + \lambdaN \sqrt{N(2\ell+1)}\|\phiproxi_{\ell}\|_1 \\
    &\le \lambdaN \sqrt{N(2\ell+1)} \|\bm{\Delta}_\ell\|_1 \le 48\qproxi_{\ell} \frac{\lambdaN^2}{\alpha_\ell}.
\end{align*}
We can get the other side of the bound using the deviation condition, the definition on $\lambdaN$, and the result in Proposition \ref{th::oracle_in}, as follows
\begin{align*}
     \|\mathbf{Y}_{\ell, I} - X_{\ell, I}\phiproxi_{\ell}\|_2^2  - \|\mathbf{Y}_{\ell, I} - X_{\ell, I}\bm{\widehat{\phi}}_{\ell, I}\|_2^2&=
-\|X_{\ell, I}\bm{\Delta}_\ell\|_2^2 + 2\bm{\Delta}_\ell'X'_{\ell, I}(\mathbf{Y}_{\ell, I} - X_{\ell, I}\phiproxi_\ell) \\
    &\le 2 \|\bm{\Delta}_\ell\|_1 \left \|X'_{\ell, I}(\mathbf{Y}_{\ell, I} - X_{\ell, I}\phiproxi_\ell) \right \|_{\infty}   \\&\le \frac{\lambdaN}{2}    \sqrt{N(2\ell+1)} \|\bm{\Delta}_\ell\|_1 \le 
 48\qproxi_{\ell} \frac{\lambdaN^2}{\alpha_\ell}.
\end{align*}

Let us consider now $\ell \in \overline{\mathcal{N}}_L$.
Starting from the definition of $\bm{\widehat{\phi}}_{\ell, I}$ and we get the following chain of inequalities
\begin{align*}
   & \|\mathbf{Y}_{\ell, I} - X_{\ell, I}\bm{\widehat{\phi}}_{\ell, I}\|_2^2 -  \left \| \mathbf{Y}_{\ell, I} - X_{\ell, I}\phiproxi_{\ell} \right\|_2^2  \le  \lambdaN \sqrt{N(2\ell+1)}\|\phiproxi_{\ell}\|_1   \\
   & \le \lambdaN \omega_\ell \sqrt{32 \max\{\qproxi_\ell, 1\} \log(p\LN)}\|\phiproxi_{\ell}\|_1 \\
   &  \le \lambdaN \omega_\ell \qproxi_{\ell} \sqrt{32C_\Phi \log(p\LN) }\\
   &\le 32p\sqrt{C_\Phi} \qproxi_\ell \frac{\lambdaN^2}{\alpha_\ell},
\end{align*}
where for the last inequality we used the fact that $\omega_\ell \sqrt{\log (pL) } \le \sqrt{32}p \lambdaN/\alpha_\ell$.
On the other hand, since the segment $I$ contains no true change point, it holds that
$$
\left \| \mathbf{Y}_{\ell, I} - X_{\ell, I}\phiproxi_{\ell} \right\|_2^2 = \|\mathbf{E}_{\ell, I} \|_2^2;
$$
moreover, for any $\eta > 0$,
\begin{equation}
    P\left ( \| \mathbf{E}_{\ell, I} \|^2_2 > \mathbb{E} \| \mathbf{E}_{\ell, I} \|^2_2 + \maxCZ  N(2\ell+1) \eta   \right) \le     P\left (\left | \| \mathbf{E}_{\ell, I} \|^2_2 - \mathbb{E} \| \mathbf{E}_{\ell, I} \|^2_2 \right |>  \maxCZ  N(2\ell+1) \eta  \right)
\end{equation}
with $\mathbb{E} \| \mathbf{E}_{\ell, I} \|^2_2  = \maxCZ N(2\ell+1)$. By adapting the deviation bounds with $\eta=\eta_\ell = c_0 \frac{\log(p\LN)}{N(2\ell+1)} \ge 1$ and using the constraint on $\lambdaN$, we obtain
\begin{align*}
     &\|\mathbf{Y}_{\ell, I} - X_{\ell, I}\phiproxi_{\ell}\|_2^2  - \|\mathbf{Y}_{\ell, I} - X_{\ell, I}\bm{\widehat{\phi}}_{\ell, I}\|_2^2  \le  \| \mathbf{E}_{\ell, I} \|_2^2 \\ &\le 2c_0 \maxCZ \log(pL) \le  \lambdaN \sqrt{\log(pL)},
\end{align*}
with probability at least $1 - c_1 e^{-c_2 \log(pL)}$, for some absolute constants $c_1, c_2>0$.

The just obtained results lead to 
\begin{align*}
& \left |\sum_{\ell=0}^{\LN-1}\| \mathbf{Y}_{\ell, I} - X_{\ell, I}\bm{\widehat{\phi}}_{\ell, I}\|_2^2 -   \sum_{\ell=0}^{\LN-1} \| \mathbf{Y}_{\ell, I} - X_{\ell, I}\phiproxi_{\ell} \|_2^2 \right | \nonumber \\&\le  \max \{48, 32p\} \max\{C_\Phi, 1\}
\sum_{\ell=0}^{\LN-1} \max \{\qproxi_\ell,1\} \frac{\lambdaN^2}{\alpha_\ell}.
\end{align*}
To conclude this proof note moreover that, for any set of coefficients $\{\bm{\beta}_\ell, \ell = 0, \dots, \LN-1\}, \bm{\beta}_\ell \in \mathbb{R}^p$, we can define an upper bound to the prediction error obtained using these coefficients 
\begin{align*}
& \sum_{\ell=0}^{\LN-1} \| \mathbf{Y}_{\ell, I} - X_{\ell, I}\phiproxi_{\ell} \|_2^2 - \sum_{\ell=0}^{\LN-1}\| \mathbf{Y}_{\ell, I} - X_{\ell, I}\bm{\beta}_\ell\|_2^2 \nonumber \\&\le  \max \{48, 32p\} \max\{C_\Phi, 1\}
\sum_{\ell=0}^{\LN-1} \max \{\qproxi_\ell,1\} \frac{\lambdaN^2}{\alpha_\ell}.
\end{align*}

\subsection{Proof of Proposition \ref{pr:4cases}} 

\subsubsection{Case 1}  We prove by contradiction assuming that
% $$\min \{|I_1|,|I_2|\} > c_4\left( \frac{\gamma + 48 \lambdaN^2 \sum_{\ell=0}^{\LN-1} \frac{\sumq}{\alpha_{\ell}}}{\kappa_\LN}\right).$$
$$\min \{|I_1|,|I_2|\} > C_\epsilon \left( \frac{\gamma+ 3\brutti}{\kappa_\LN}\right).$$

We first observe that
%\begin{align*}
%    &\min \{|I_1|,|I_2|\} > c_4\left( \frac{\gamma + 48 \lambdaN^2 \sum_{\ell=0}^{\LN-1} \frac{\sumq}{\alpha_{\ell}}}{\kappa_\LN}\right) \ge c_4\left( \frac{\gamma}{\kappa_\LN} + \frac{ 48 \lambdaN^2 \sumq \LN}{4 C_{\Phi}\alpha^2_\LN \sum_{\ell = 0}^{\LN-1}(2\ell+1)}\right)\\
%    & \ge c_4\left( \frac{\gamma}{\kappa_\LN} + \frac{48\cdot 4\cdot 10}{C_{\Phi}} \frac{\FN^2}{\alpha^2_\LN} \frac{\log(p\LN) p\LN}{\sum_{\ell = 0}^{\LN-1}(2\ell+1)}\right)
%\end{align*}
\begin{align} \label{eq:contr_12}
    &\min \{|I_1|,|I_2|\} > C_\epsilon \left( \frac{\gamma+ 3\brutti}{\kappa_\LN}\right) > 32 \,\underset{\ell < L}{\max} \left\{\frac{\omega_\ell^2}{2\ell+1} \right\}\qN \log(p\LN),
\end{align}
since $\kappa_\LN \le 4 C_{\Phi} \sum_{\ell = 0}^{\LN-1}(2\ell+1) \alpha_\ell $ and
$$\gamma \ge 4 \brutti, \qquad \lambdaN \ge 4 \FN \sqrt{\log p\LN} \cdot \sqrt{\frac{32p^2 \sum_{\ell=0}^{\LN-1} (2\ell+1)\alpha_\ell}{\alpha_\ell}}.
$$
Now, since $|I_j| \ge 32 \,\underset{\ell < L}{\max} \left\{\frac{\omega_\ell^2}{2\ell+1} \right\}\qN \log(p\LN)$, $j =1,2$, the deviation condition and the compatibility condition (Propositions \ref{prop::dev_condition} and \ref{prop::re_compatibility}) are satisfied for both intervals with probability at least $1 - c_1 e^{-c_2 \log(pL)}$ for some absolute constants $c_1, c_2>0$.

Without loss of generality, we assume $I_1 \subset [\eta_1, \eta_2)$ and $I_2 \subset [\eta_2, \eta_3)$.
It follows from \eqref{eq:cond_lemma12} and \eqref{eq:prederr_phihat} in Theorem \ref{th::error_diff} that, with probability at least $1 - c_1 e^{-c_2 \log(pL)}$, for some absolute constants $c_1, c_2>0$,
\begin{align*}
    &\sum_{\ell = 0}^{\LN-1}\mathcal{L}_{\ell}(I_1, \phihat_{\ell,I}) + \sum_{\ell = 0}^{\LN-1}\mathcal{L}_{\ell}(I_2, \phihat_{\ell,I}) \le \sum_{\ell = 0}^{\LN-1}\mathcal{L}_{\ell}(I, \phihat_{\ell,I})\\
    &\le \sum_{\ell = 0}^{\LN-1}\mathcal{L}_{\ell}(I_1, \phihat_{\ell,I_1}) + \sum_{\ell = 0}^{\LN-1}\mathcal{L}_{\ell}(I_2, \phihat_{\ell,I_2}) + \gamma \\
    & \le  \sum_{\ell = 0}^{\LN-1}\mathcal{L}_{\ell}(I_1, \bm{\phi}_{\ell}^{(1)}) + \sum_{\ell = 0}^{\LN-1}\mathcal{L}_{\ell}(I_2, \bm{\phi}_{\ell}^{(2)}) + \gamma  + 48  \sum_{\ell=0}^{\LN-1} \lambdaN^2 \frac{q^{(1)}_{\ell}}{\alpha_{\ell}} + 48  \sum_{\ell=0}^{\LN-1} \lambdaN^2\frac{q^{(2)}_{\ell}}{\alpha_{\ell}}.
\end{align*}
%where $\phiproxi_{\ell, I_j}$ is defined as in \eqref{eq:phiproxi}, considering the set $I_j$ instead of the entire interval $I$, for $j =1,2$ and $\ell = 0, \dots, \LN-1$.

Let us define $\bm{\Delta}_{\ell, j} = \phihat_{\ell,I} - \bm{\phi}_{\ell}^{(j)}$, for $j =1,2$, and using similar computation to \eqref{eq:proof_res}, we get
\begin{align}
& \sum_{\ell = 0}^{\LN-1}|I_1|(2\ell+1)\bm{\Delta}'_{\ell, 1}\widehat{\Gamma}_{\ell, I_1}\bm{\Delta}_{\ell, 1} +  \sum_{\ell = 0}^{\LN-1}|I_2|(2\ell+1)\bm{\Delta}'_{\ell, 2}\widehat{\Gamma}_{\ell, I_2}\bm{\Delta}_{\ell, 2}\nonumber \\ 
& \le 2\sum_{\ell = 0}^{\LN-1} \|\bm{\Delta}_{\ell, 1}\|_1 \FN \sqrt{|I_1|(2\ell+1)\log(p\LN)} +  2 \sum_{\ell = 0}^{\LN-1}\|\bm{\Delta}_{\ell, 2}\|_1 \FN \sqrt{|I_2|(2\ell+1)\log(p\LN)}  \nonumber \\ 
&+ \gamma  + 48 \sum_{\ell=0}^{\LN-1}  \lambdaN^2 \frac{q_{\ell}^{(1)}+ q_{\ell}^{(2)}}{\alpha_{\ell}}\nonumber \\ 
& \le \sum_{\ell = 0}^{\LN-1} \frac{\lambdaN}{2}\|\bm{\Delta}_{\ell, 1}\|_1 \sqrt{|I_1|(2\ell+1)} + \sum_{\ell = 0}^{\LN-1} \frac{\lambdaN}{2}\|\bm{\Delta}_{\ell, 2}\|_1 \sqrt{|I_2|(2\ell+1)}\nonumber \\ 
&+ \gamma  + 48  \sum_{\ell=0}^{\LN-1} \lambdaN^2 \frac{q_{\ell}^{(1)}+ q_{\ell}^{(2)}}{\alpha_{\ell}}. \label{eq:gamma_scomp}
\end{align} 
Since
\begin{equation}\label{eq:delta_scomp}
    \|\bm{\Delta}_{\ell, j}\|_1 \le  \sqrt{p}\|\bm{\Delta}_{\ell, j}\|_2, \quad j=1,2,
\end{equation}
we have that
\begin{align}
    \eqref{eq:gamma_scomp} &\le \sum_{\ell = 0}^{\LN-1} \frac{\lambdaN}{2}\sqrt{p}\|\bm{\Delta}_{\ell, 1}\|_2 \sqrt{|I_1|(2\ell+1)} + \sum_{\ell = 0}^{\LN-1} \frac{\lambdaN}{2}\sqrt{p}\|\bm{\Delta}_{\ell, 2}\|_2 \sqrt{|I_2|(2\ell+1)}\nonumber \\ 
& + \gamma  + 48 \sum_{\ell=0}^{\LN-1} \lambdaN^2 \frac{q_{\ell}^{(1)}+ q_{\ell}^{(2)}}{\alpha_{\ell}} \nonumber\\
    & \le   \frac{1}{256} \sum_{\ell = 0}^{\LN-1} \alpha_{\ell} |I_1|(2\ell+1)\|\bm{\Delta}_{\ell, 1}\|^2_2 + \frac{1}{256} \sum_{\ell = 0}^{\LN-1} \alpha_{\ell} |I_2|(2\ell+1)\|\bm{\Delta}_{\ell, 2}\|^2_2 \nonumber \\ 
&+ \gamma  + 48 \sum_{\ell=0}^{\LN-1}  \lambdaN^2 \frac{q_{\ell}^{(1)}+ q_{\ell}^{(2)}}{\alpha_{\ell}} + 32p \sum_{\ell = 0}^{\LN-1}    \frac{\lambdaN^2}{\alpha_{\ell}} \nonumber\\
    & \le  \frac{1}{256} \sum_{\ell = 0}^{\LN-1} \alpha_{\ell} |I_1|(2\ell+1)\|\bm{\Delta}_{\ell, 1}\|^2_2 + \frac{1}{256} \sum_{\ell = 0}^{\LN-1} \alpha_{\ell} |I_2|(2\ell+1)\|\bm{\Delta}_{\ell, 2}\|^2_2 \nonumber \\ 
&+ \gamma  + 3\brutti, \label{eq:re_gamma}
\end{align}
where the second inequality is obtained using $2a_\ell b_\ell \le a_\ell^2 + b_\ell^2$ with
$$
a_\ell = 4 \lambdaN\frac{\sqrt{p}}{\sqrt{\alpha_{\ell}}}, \qquad b_\ell =\sqrt{ \alpha_\ell |I_j|(2\ell +1)}\|\bm{\Delta}_{\ell, j}\|_2/16.
$$
From the Compatibility Condition (Propositions \ref{prop::re_compatibility}) it holds
\begin{align}\label{eq:re_j}
 \sum_{\ell = 0}^{\LN-1}|I_j|(2\ell+1)\alpha_{\ell} \|\bm{\Delta}_{\ell, j}\|^2_2\le \sum_{\ell = 0}^{\LN-1}|I_j|(2\ell+1)\bm{\Delta}'_{\ell, 1}\widehat{\Gamma}_{\ell, I_j}\bm{\Delta}_{\ell, j} + \sum_{\ell = 0}^{\LN-1}|I_j|(2\ell+1)\tau_{\ell,j} \|\bm{\Delta}_{\ell, j}\|^2_1,
\end{align} 
where $\tau_{\ell,j} = \alpha_{\ell}  \omega_\ell^2 \frac{\log(p\LN)}{|I_j|(2\ell+1)}$, as stated in \eqref{eq:alpha_tau_def}.
Considering the last term of the just stated inequality and using \eqref{eq:delta_scomp}, we obtain
\begin{align*}
    & \sum_{\ell = 0}^{\LN-1}|I_j|(2\ell+1)\tau_{\ell,j} \|\bm{\Delta}_{\ell, j}\|^2_1 \le \sum_{\ell = 0}^{\LN-1}|I_j|(2\ell+1)\tau_{\ell,j}  p\|\bm{\Delta}_{\ell, j}\|^2_2\\
    %& = 2  \sum_{\ell = 0}^{\LN-1}|I_j|(2\ell+1)\tau_{\ell,j} \qproxi_{\ell}\|\bm{\Delta}_{\ell, j}\|^2_2 + 2 (48)^2 \sum_{\ell = 0}^{\LN-1}\tau_{\ell,j} (\qproxi_\ell)^2 \frac{\lambdaN^2}{\alpha^2_\ell }\\
    %& = 2  \sum_{\ell = 0}^{\LN-1}(2\ell+1)\alpha_{\ell}  \max \left\{  \omegaN^{2}, 1\right\} \log(p\LN) \qproxi_{\ell} \| \bm{\Delta}_{\ell, j} \|^2_2 + 2 (48)^2 \sum_{\ell = 0}^{\LN-1} \max \left\{  \omegaN^{2}, 1\right\}\frac{\log(p\LN)}{|I_j|} (\qproxi_\ell)^2 \frac{\lambdaN^2}{\alpha_\ell }\\
    &\le \frac{1}{16} |I_j| \sum_{\ell = 0}^{\LN-1}(2\ell+1)\alpha_{\ell} \| \bm{\Delta}_{\ell, j} \|^2_2 ,
    %+ \frac{ (48)^2}{16} \sum_{\ell = 0}^{\LN-1}  \qproxi_\ell \frac{\lambdaN^2}{\alpha_\ell },
\end{align*}
since $p\tau_{\ell,j}\le \alpha_{\ell} /32$. It follows that
\begin{align}
    & \sum_{\ell = 0}^{\LN-1}|I_1|(2\ell+1)\tau_{\ell,1} \|\bm{\Delta}_{\ell, 1}\|^2_1 + \sum_{\ell = 0}^{\LN-1}|I_2|(2\ell+1)\tau_{\ell,2} \|\bm{\Delta}_{\ell, 2}\|^2_1\nonumber\\
    &\le \frac{1}{16} |I_1| \sum_{\ell = 0}^{\LN-1}(2\ell+1)\alpha_{\ell} \| \bm{\Delta}_{\ell, 1} \|^2_2  + \frac{1}{16} |I_2| \sum_{\ell = 0}^{\LN-1}(2\ell+1)\alpha_{\ell} \| \bm{\Delta}_{\ell, 2} \|^2_2.\label{eq:re_tau}
\end{align}
From \eqref{eq:re_j}, thanks to the results in \eqref{eq:re_gamma} and \eqref{eq:re_tau}, we obtain the following upper bound
\begin{align*}
 &\sum_{\ell = 0}^{\LN-1}|I_1|(2\ell+1)\alpha_{\ell} \|\bm{\Delta}_{\ell, 1}\|^2_2 + \sum_{\ell = 0}^{\LN-1}|I_1|(2\ell+1)\alpha_{\ell} \|\bm{\Delta}_{\ell, 1}\|^2_2  \\
 &\le \frac{1}{256} \sum_{\ell = 0}^{\LN-1} \alpha_{\ell} |I_1|(2\ell+1)\|\bm{\Delta}_{\ell, 1}\|^2_2 + \frac{1}{256} \sum_{\ell = 0}^{\LN-1} \alpha_{\ell} |I_2|(2\ell+1)\|\bm{\Delta}_{\ell, 2}\|^2_2\\ 
 & + \frac{1}{16} |I_1| \sum_{\ell = 0}^{\LN-1}(2\ell+1)\alpha_{\ell} \| \bm{\Delta}_{\ell, 1} \|^2_2 + \frac{1}{16} |I_2| \sum_{\ell = 0}^{\LN-1}(2\ell+1)\alpha_{\ell} \| \bm{\Delta}_{\ell, 2} \|^2_2\\ 
 &+ \gamma  + 3\brutti  \\
&= \frac{17}{256} \sum_{\ell = 0}^{\LN-1} \alpha_{\ell} |I_1|(2\ell+1)\|\bm{\Delta}_{\ell, 1}\|^2_2 + \frac{17}{256}  \sum_{\ell = 0}^{\LN-1} \alpha_{\ell} |I_2|(2\ell+1)\|\bm{\Delta}_{\ell, 2}\|^2_2 \\ 
 & + \gamma + 3\brutti.
\end{align*}
The above equation can be written as
\begin{align*}
    \frac{239}{256}\sum_{\ell = 0}^{\LN-1}|I_1|(2\ell+1)\alpha_{\ell} \|\bm{\Delta}_{\ell, 1}\|^2_2 &+ \frac{239}{256}\sum_{\ell = 0}^{\LN-1}|I_2|(2\ell+1)\alpha_{\ell} \|\bm{\Delta}_{\ell, 2}\|^2_2 \le  \gamma + 3\brutti.
\end{align*}
The lower bound is obtained as follows
\begin{align*}
     & \sum_{\ell = 0}^{\LN-1}|I_1|(2\ell+1)\alpha_{\ell} \|\bm{\Delta}_{\ell, 1}\|^2_2 +  \sum_{\ell = 0}^{\LN-1}|I_2|(2\ell+1)\alpha_{\ell} \|\bm{\Delta}_{\ell, 2}\|^2_2 \\
     &\ge \sum_{\ell = 0}^{\LN-1}\underset{v_{\ell} \in \mathbb{R}^p}{\inf}\left( |I_1| \| \bm{\phi}_\ell^{(1)}-v_\ell\|_2^2 + |I_2| \| \bm{\phi}_\ell^{(2)}-v_\ell\|_2^2\right)(2\ell+1)\alpha_{\ell} \\
     & \ge \frac{|I_1||I_2|}{|I|} \sum_{\ell = 0}^{\LN-1}(2\ell+1)\alpha_{\ell} \| \bm{\phi}_\ell^{(2)} - \bm{\phi}_\ell^{(1)} \|^2_2 \\
     & \ge \frac{\min \{|I_1|,|I_2|\}}{2} \sum_{\ell = 0}^{\LN-1}(2\ell+1)\alpha_{\ell} \| \bm{\phi}_\ell^{(2)} - \bm{\phi}_\ell^{(1)} \|^2_2 .
\end{align*}
Combining the upper and lower bounds we get
%\begin{align*}
%&\frac{239}{256}\frac{\min \{|I_1|,|I_2|\}}{2} \sum_{\ell = 0}^{\LN-1}(2\ell+1)\alpha_{\ell} \| \bm{\phi}_\ell(\eta) - \bm{\phi}_\ell(\eta-1) \|^2_2 \le \gamma + 48 \lambdaN^2 \sum_{\ell=0}^{\LN-1} \frac{q_{\ell}^{(1)}+ q_{\ell}^{(2)}+8\qproxi_{\ell}}{\alpha_{\ell}}\\
%&\min \{|I_1|,|I_2|\} \sum_{\ell = 0}^{\LN-1}(2\ell+1)\alpha_{\ell} \| \bm{\phi}_\ell(\eta) - \bm{\phi}_\ell(\eta-1) \|^2_2 \le \frac{2\cdot256}{239}\left( \gamma + 48 \lambdaN^2 \sum_{\ell=0}^{\LN-1} \frac{q_{\ell}^{(1)}+ q_{\ell}^{(2)}+8\qproxi_{\ell}}{\alpha_{\ell}}\right)\\
%&\min \{|I_1|,|I_2|\} \le \frac{2\cdot256}{239}\left( \frac{\gamma + 48 \lambdaN^2 \sum_{\ell=0}^{\LN-1} \frac{q_{\ell}^{(1)}+ q_{\ell}^{(2)}+8\qproxi_{\ell}}{\alpha_{\ell}}}{\sum_{\ell = 0}^{\LN-1}(2\ell+1)\alpha_{\ell} \| \bm{\phi}_\ell(\eta) - \bm{\phi}_\ell(\eta-1) \|^2_2}\right)\\
%&\min \{|I_1|,|I_2|\} \le c_4\left( \frac{\gamma + 48 \lambdaN^2 \sum_{\ell=0}^{\LN-1} \frac{\sumq}{\alpha_{\ell}}}{\kappa_\LN}\right),
%\end{align*}

\begin{align*}
&\frac{239}{256}\frac{\min \{|I_1|,|I_2|\}}{2} \sum_{\ell = 0}^{\LN-1}(2\ell+1)\alpha_{\ell} \| \bm{\phi}_\ell^{(2)} - \bm{\phi}_\ell^{(1)} \|^2_2 \le \gamma + 3\brutti,\\
&\min \{|I_1|,|I_2|\} \sum_{\ell = 0}^{\LN-1}(2\ell+1)\alpha_{\ell} \| \bm{\phi}_\ell^{(2)} - \bm{\phi}_\ell^{(1)} \|^2_2 \le \frac{2\cdot256}{239}\left( \gamma + 3\brutti\right),\\
%&\min \{|I_1|,|I_2|\} \le \frac{2\cdot256}{239}\left( \frac{\gamma + 3\brutti}{\sum_{\ell = 0}^{\LN-1}(2\ell+1)\alpha_{\ell} \| \bm{\phi}_\ell(\eta) - \bm{\phi}_\ell(\eta-1) \|^2_2}\right)\\
%&\min \{|I_1|,|I_2|\} \le \frac{2\cdot256}{239}\left( \frac{\gamma + 3\brutti}{\kappa_\LN}\right),\\
&\min \{|I_1|,|I_2|\} \le C_\epsilon \left( \frac{\gamma+ 3\brutti}{\kappa_\LN}\right),
\end{align*}
setting $C_\epsilon = 3$ and $\kappa_\LN = \min_{k=0,\dots,K}\sum_{\ell = 0}^{\LN-1}(2\ell+1)\alpha_{\ell} \| \bm{\phi}_\ell^{(k+1)} - \bm{\phi}_\ell^{(k)} \|^2_2 $. This contradicts \eqref{eq:contr_12}, proving the proposition.\\

\subsubsection{Case 2}  Without loss of generality, we assume $I_1 \subset [\eta_1, \eta_2)$, $I_2 = [\eta_2, \eta_3)$  and $I_3 \subset [\eta_3, \eta_4)$.

We need to prove that 
$$
|I_1| \le C_\epsilon \left( \frac{2\gamma+ 5\brutti }{\kappa_\LN}\right),
$$
and we proceed by contradiction, assuming that
\begin{equation} \label{eq:contr_13}
    |I_1| > C_\epsilon \left( \frac{2\gamma+ 5\brutti}{\kappa_\LN}\right)>32\, \underset{\ell<L}{\max} \left\{\frac{\omega_\ell}{2\ell+1} \right\} \qN \log(p\LN).
\end{equation}
We have then
\begin{align*}
    &\sum_{\ell = 0}^{\LN-1}\mathcal{L}_{\ell}(I_1, \phihat_{\ell,I}) + \sum_{\ell = 0}^{\LN-1}\mathcal{L}_{\ell}(I_2, \phihat_{\ell,I}) + \sum_{\ell = 0}^{\LN-1}\mathcal{L}_{\ell}(I_3, \phihat_{\ell,I})
    \le   \sum_{\ell = 0}^{\LN-1}\mathcal{L}_{\ell}(I, \phihat_{\ell,I}) \\&\le \sum_{\ell = 0}^{\LN-1}\mathcal{L}_{\ell}(I_1, \phihat_{\ell,I_1})  + \sum_{\ell = 0}^{\LN-1}\mathcal{L}_{\ell}(I_2, \phihat_{\ell,I_2}) + \sum_{\ell = 0}^{\LN-1}\mathcal{L}_{\ell}(I_3, \phihat_{\ell,I_3}) + 2\gamma\\
    &\le\sum_{\ell = 0}^{\LN-1}\mathcal{L}_{\ell}(I_1, \bm{\phi}_{\ell}^{(1)})  + \sum_{\ell = 0}^{\LN-1}\mathcal{L}_{\ell}(I_2, \bm{\phi}_{\ell}^{(2)}) + \sum_{\ell = 0}^{\LN-1}\mathcal{L}_{\ell}(I_3, \bm{\phi}_{\ell}^{(3)}) + 2\gamma + 48  \sum_{\ell=0}^{\LN-1} \lambdaN^2\frac{q^{(1)}_{\ell} + q^{(2)}_{\ell}}{\alpha_{\ell}} +\brutti,
\end{align*}
which implies
\begin{align*}
    &\sum_{\ell = 0}^{\LN-1}\mathcal{L}_{\ell}(I_1, \phihat_{\ell,I}) + \sum_{\ell = 0}^{\LN-1}\mathcal{L}_{\ell}(I_2, \phihat_{\ell,I}) \\
    &\le \sum_{\ell = 0}^{\LN-1}\mathcal{L}_{\ell}(I_1, \bm{\phi}_{\ell}^{(1)})  + \sum_{\ell = 0}^{\LN-1}\mathcal{L}_{\ell}(I_2, \bm{\phi}_{\ell}^{(2)})  + 2\gamma + 48 \sum_{\ell=0}^{\LN-1} \lambdaN^2\frac{q^{(1)}_{\ell} + q^{(2)}_{\ell}}{\alpha_{\ell}}+2\brutti. 
\end{align*}

It follows from identical arguments in the proof of \textbf{Case 1} of Proposition \ref{pr:4cases} that 
$$
\min \{|I_1|,|I_2|\} \le C_\epsilon \left( \frac{2\gamma+ 5\brutti}{\kappa_\LN}\right). 
$$
Since $|I_2|>\Delta$ by definition (see \eqref{eq:delta}), it follows from Condition \ref{cond:snr} that
$$
|I_1| \le C_\epsilon \left( \frac{2\gamma+5\brutti}{\kappa_\LN}\right), 
$$
which contradicts \eqref{eq:contr_13}.

\subsubsection{Case 3}  Without loss of generality we assume $I \subset [\eta_1, \eta_2)$. It follows from Theorem \ref{th::error_diff} that, with probability at least $1 - c_1 e^{-c_2 \log(pL)}$, for some absolute constants $c_1, c_2>0$,
\begin{align*}
&\max_{J \in \{[s,b),[b, e],I\}} \left |\sum_{\ell = 0}^{\LN-1}\mathcal{L}_{\ell}(J, \phihat_{\ell,J}) - \sum_{\ell = 0}^{\LN-1}\mathcal{L}_{\ell}(J, \bm{\phi}_\ell^{(1)}) \right | \le   \brutti < \gamma/4
\end{align*}

Let $B = [b-p, b+p-1]$,
$$
\sum_{\ell = 0}^{\LN-1}\mathcal{L}_{\ell}(B, \bm{\phi}_\ell^{(1)}) = \sum_{\ell = 0}^{\LN-1} \|\mathbf{E}_{\ell,B}\|_2^2 \le \brutti.
$$
Then, the final claim holds due to the choice of $\gamma$.\\

\subsubsection{Case 4} 

Without loss of generality, we assume 
\(I_1 = [\eta_1, \eta_2)\) and \(I_2 = [\eta_2, \eta_3)\). Choosing \(\lambda_\ell \ge 4 \, F_N \sqrt{\log(p \, \LN)}\), for \(\ell = 0, \dots, \LN-1\), allows us to apply Proposition \ref{th::oracle_in} and Theorem \ref{th::error_diff}.

By the definition of \(\Delta\) (see \eqref{eq:delta}), Condition \ref{cond:snr}, and the choice of the penalty parameters, it follows that, for any \(j \in \{1, \dots, J-1\}\),

\begin{equation}
    |I_j| \ge \Delta \ge  C_\epsilon \left( \frac{J\gamma+ (2J+1)\brutti}{\kappa_\LN}\right) > 32 \underset{\ell < L}{\max} \left\{\frac{\omega_\ell^2}{2\ell+1} \right\} \qN \log(p\LN). \label{eq:delta_num}
\end{equation}
The proposition is proved by contradiction by showing that \eqref{eq:delta_num} does not hold for at least one $j \in \{1, \dots, J-1\}$. We will focus our attention on $j=1,2$. 
Hence, assume that
\begin{equation*}
    \sum_{\ell = 0}^{\LN-1}\mathcal{L}_{\ell}(I, \phihat_{\ell,I}) \le \sum_{j=0}^{J}\sum_{\ell = 0}^{\LN-1}\mathcal{L}_{\ell}(I_j, \phihat_{\ell,I_j}) + J\gamma.
\end{equation*}

Let us define $\bm{\Delta}_{\ell, j} = \phihat_{\ell,I} - \bm{\phi}_{\ell}^{(j)}$, for $j =0,\dots, J.$ It follows from Theorem \ref{th::error_diff} that, with probability at least $1 - c_1 e^{-c_2 \log(pL)}$, for some absolute constants $c_1, c_2>0$,
\begin{align*}
     &\sum_{j=0}^{J}\sum_{\ell = 0}^{\LN-1}\mathcal{L}_{\ell}(I_j, \phihat_{\ell,I}) \le \sum_{\ell = 0}^{\LN-1}\mathcal{L}_{\ell}(I, \phihat_{\ell,I}) \\& \le \sum_{j=0}^{J}\sum_{\ell = 0}^{\LN-1}\mathcal{L}_{\ell}(I_j, \phihat_{\ell,I_j}) + J\gamma \\
     & \le \sum_{j=0}^{J}\sum_{\ell = 0}^{\LN-1}\mathcal{L}_{\ell}(I_j, \phi_{\ell}^{(j)}) + J\gamma + 48  \sum_{\ell=0}^{\LN-1} \lambdaN^2 \frac{q^{(1)}_{\ell} + q^{(2)}_{\ell}}{\alpha_{\ell}}\\ &+(J-1)\brutti,  
\end{align*}
%where $\mathcal{J} = \{0, 3, \dots, J\}$, 
which implies that
\begin{align}
    & \sum_{j = 1}^{2}\sum_{\ell = 0}^{\LN-1}|I_j|(2\ell+1)\bm{\Delta}'_{\ell, j}\widehat{\Gamma}_{\ell, I_j}\bm{\Delta}_{\ell, j}\nonumber \\
    & \le 2\sum_{j = 1}^{2}\sum_{\ell = 0}^{\LN-1} \|\bm{\Delta}_{\ell, j}\|_1 \FN \sqrt{|I_j|(2\ell+1)\log(p\LN)} +  J \gamma  + 48 \sum_{\ell=0}^{\LN-1} \lambdaN^2 \frac{q^{(1)}_{\ell} + q^{(2)}_{\ell}}{\alpha_{\ell}} \nonumber\\ 
    &+2(J-1) \brutti. \label{eq:step_lemma15}
\end{align}

Following the same steps of the proof of \textbf{Case 2} of Proposition \ref{pr:4cases}, \eqref{eq:step_lemma15} leads to the following result 
$$
\min \{|I_1|,|I_2|\} \le C_\epsilon \left( \frac{J\gamma+ (2J+1)\brutti}{\kappa_\LN}\right),
$$
which contradicts \eqref{eq:delta_num} and proves the proposition.

\subsection{Proof of Proposition \ref{pr:k=k}}
Consider the true partition 
$$
I_k = [\eta_k, \eta_{k+1}) \qquad k = 0, \dots, K,
$$
and the estimated partition
$$
\hat{I}_k = [\hat{\eta}_k, \hat{\eta}_{k+1}) \qquad k = 0, \dots, \hat{K},
$$
with estimated change points $1=\hat{\eta}_0 < \hat{\eta}_1 < \dots < \hat{\eta}_{\hat{K}} < n < \hat{\eta}_{\hat{K}+1}=n+1$.

Define moreover the sequence $\{w_k, k = 0, \dots, D+1\}$ obtained by sorting the set $\{\eta_k\}_{k=1}^K \cup \{\hat{\eta}_k\}_{k=0}^{\hat{K}+1}$. Note that $D\le K + \hat K$. 
Consider the partition
$$
W_k = [w_k,w_{k+1}), \qquad k = 0, \dots, D.
$$
We want to prove the following chain of inequalities
\begin{align}
    \sum_{k=0}^K\sum_{\ell = 0}^{\LN-1}\mathcal{L}_\ell(I_k, \bm{\phi}_\ell^{(k)}) + K\gamma &\ge \sum_{k=0}^K\sum_{\ell = 0}^{\LN-1}\mathcal{L}_\ell(I_k, \phihat_{\ell, I_k}) + K\gamma - (K+1)\brutti \label{eq:chain1} \\ 
    & \ge \sum_{k=0}^{\hat{K}}\sum_{\ell = 0}^{\LN-1}\mathcal{L}_\ell(\hat{I}_k, \phihat_{\ell, \hat{I}_k}) + \hat{K}\gamma - (K+1)\brutti \label{eq:chain2} \\
    & \ge \sum_{k=0}^{D}\sum_{\ell = 0}^{\LN-1}\mathcal{L}_\ell(W_k, \phihat_{\ell,W_k}) + \hat{K}\gamma - 3(K+1)\brutti\label{eq:chain3} 
\end{align}
and that 
\begin{align}
    \left | \sum_{k=0}^K\sum_{\ell = 0}^{\LN-1}\mathcal{L}_\ell(I_k, \bm{\phi}^{(k)}_{\ell}) - \sum_{k=0}^{D}\sum_{\ell = 0}^{\LN-1}\mathcal{L}_\ell(W_k, \phihat_{\ell,W_k}) \right|  \le (2\hat{K}+K+1)\brutti
\end{align}
If so, it must hold that $|\hat{\mathcal{P}}|=K$, as otherwise if $\hat{K}\ge K +1$, then
\begin{align} 
     (2\hat{K}+K+1)\brutti \ge (\hat{K}-K)\gamma - 3(K+1)\brutti \ge \gamma - 3(K+1)\brutti,
\end{align}
and then
$$
(10K+4)\brutti \ge \gamma,
$$
since $\hat{K} \le 3K$. This contradicts the choice of $\gamma.$

Note that \eqref{eq:chain1} is a direct consequence of Theorem \ref{th::error_diff} since
\begin{align}
    &\left | \sum_{k=0}^K\sum_{\ell = 0}^{\LN-1}\mathcal{L}_\ell(I_k, \bm{\phi}_\ell^{(k)}) - \sum_{k=0}^K\sum_{\ell = 0}^{\LN-1}\mathcal{L}_\ell(I_k, \phihat_{\ell, I_k}) \right | \le (K+1)\brutti, 
\end{align}
and \eqref{eq:chain2} follows from $\hat{\mathcal{P}}$ being a minimizer of \eqref{eq:min_problem}. \\

Let us consider the first estimated interval $\hat{I}_0$ and assume it contains $q$ true change points $\eta_1, \dots, \eta_q$ which have not been detected, that is $\hat{I}_0 = \bigcup_{k=0}^q W_k$. 

Using \eqref{eq:prederr_phihat} from Theorem \ref{th::error_diff} we get
\begin{align*}
    & \sum_{k = 0}^q \mathcal{L}_\ell(W_k,\phihat_{\ell, W_k}) \le \sum_{k = 0}^q \mathcal{L}_\ell(W_k,\bm{\phi}^{(k)}_{\ell})+ (q+1)\brutti.
\end{align*}
Using \eqref{eq:prederr_beta} from Theorem \ref{th::error_diff}, since we are considering the coefficients $\{\phihat_{\ell, \hat{I}_0}\}$ estimated on the interval $\hat{I}_0$ instead of the segments $W_k, k= 0, \dots,q$, we get
\begin{align*}
    & \sum_{k = 0}^q \mathcal{L}_\ell(W_k,\phihat_{\ell, \hat{I}_0}) \ge \sum_{k = 0}^q \mathcal{L}_\ell(W_k,\bm{\phi}^{(k)}_{\ell})
    -  (q+1)\brutti.
\end{align*}
So we have that for the fixed interval $\hat{I}_0$
\begin{align*}
    & \sum_{k = 0}^q \mathcal{L}_\ell(W_k,\phihat_{\ell, \hat{I}_0}) \ge \sum_{k = 0}^q \mathcal{L}_\ell(W_k,\phihat_{\ell, W_k})
    - 2 (q+1)\brutti
\end{align*}
which leads to \eqref{eq:chain3}. \\

For the fixed interval $I_0$ containing $r$ estimated change points $\hat{\eta}_1, \dots, \hat{\eta}_r$, i.e., $I_0 = \bigcup_{k=0}^r W_k$, note that
\begin{align}
    &\left| \sum_{\ell = 0}^{\LN-1}\mathcal{L}_\ell(I_0, \bm{\phi}^{(0)}_{\ell}) - \sum_{k=0}^{r}\sum_{\ell = 0}^{\LN-1}\mathcal{L}_\ell(W_k, \phihat_{\ell,W_k})\right | \\ 
    & \le \left|\sum_{k=0}^r\sum_{\ell = 0}^{\LN-1}\mathcal{L}_\ell(W_k, \bm{\phi}^{(0)}_{\ell}) - \sum_{k=0}^{r}\sum_{\ell = 0}^{\LN-1}\mathcal{L}_\ell(W_k, \phihat_{\ell,W_k})\right | + \sum_{k=1}^{r}\sum_{\ell = 0}^{\LN -1} \| \mathbf{E}_{\ell, B_k}\|_2^2\\ 
    & \le  (2r+1)\brutti
\end{align}
where $B_k = [\hat{\eta}_{k} - p, \hat{\eta}_k + p -1]$.

%%%%%%%%%%%%%%%%%%

\end{document}